\newcommand{\extended}[1]{}    
\newcommand{\short}[1]{#1}     
\renewcommand{\extended}[1]{#1} 
\renewcommand{\short}[1]{}      

\PassOptionsToPackage{usenames,dvipsnames}{xcolor}

\documentclass[sigconf]{aamas}

\usepackage{balance}
\usepackage{amsmath}
\usepackage{stmaryrd}
\usepackage{xcolor}
\usepackage{xfrac,xspace}
\usepackage{graphicx}
\usepackage{wrapfig}
\usepackage{marginnote}
\usepackage[normalem]{ulem}
\usepackage[super]{nth}
\usepackage{breakcites}
\usepackage[fillcomment,nosemicolon,ruled,vlined]{algorithm2e}
\usepackage[nameinlink]{cleveref}
\usepackage{etoolbox}
\usepackage{lscape}
\usepackage[inline,shortlabels]{enumitem}
\setlist[itemize]{leftmargin=*}
\usepackage{rotating}
\usepackage{makecell}
\usepackage{array}
\newcolumntype{h}{>{\setbox0=\hbox\bgroup}c<{\egroup}@{}}
\usepackage{colortbl}
\usepackage{multirow}
\usepackage{tablefootnote}

\crefname{algorithm}{alg.}{algs.}
\Crefname{algorithm}{Algorithm}{Algorithms}
\crefname{table}{tab.}{tabs.}
\Crefname{table}{Table}{Tables}
\crefname{section}{Section}{Sections}
\Crefname{section}{Section}{Sections}

\graphicspath{{figures/},{imgs/}}

\usepackage{tikz}
\usetikzlibrary{automata, arrows.meta, positioning,shapes.multipart}

\colorlet{colorAttribute}{PineGreen!75!green!85!black}
\colorlet{colorCondition}{RedOrange}
\colorlet{colorTransition}{Blue}

\newcommand{\eg}{e.g.\xspace}
\newcommand{\ie}{i.e.\xspace}
\newcommand{\st}{s.t.\xspace}

\Crefname{section}{Sec.}{Secs.}
\crefname{section}{section}{sections}
\Crefname{subsection}{Sec.}{Secs.}
\crefname{subsection}{section}{sections}
\Crefname{equation}{Eq.}{Eqs.}
\crefname{equation}{equation}{equations}
\Crefname{definition}{Def.}{Defs.}
\crefname{definition}{definition}{definitions}
\Crefname{algorithm}{Alg.}{Algs.}
\crefname{algorithm}{algorithm}{algorithms}
\crefname{table}{table}{tables}
\Crefname{figure}{Fig.}{Figs.}
\crefname{figure}{figure}{figures}
\Crefname{example}{Ex.}{Exs.}
\crefname{example}{example}{examples}

\makeatletter
\renewcommand{\paragraph}{\@startsection{paragraph}{4}{0pt}%
  {.8ex plus 0.2ex minus 0.2ex}%
  {-0.5em}%
  {\bfseries}}
\makeatother

\makeatletter
\newcommand\dontbreakpar{\par\nobreak\@afterheading}
\makeatother

\def\event/{action}
\def\Event/{Action}
\newcommand{\set}[1]{\{{#1}\}}
\newcommand{\Reals}{\mathbb{R}}

\newcommand{\fpart}{\rightharpoonup}
\newcommand{\para}[1]{\smallskip\noindent\textbf{#1}}

\newcommand{\MAS}{MAS\xspace}
\newcommand{\AMAS}{AMAS\xspace}
\newcommand{\DAMAS}{DAMAS\xspace}
\newcommand{\DMAS}{DMAS\xspace}
\newcommand{\TAMAS}{CAMAS\xspace}
\newcommand{\TMAS}{CMAS\xspace}
\newcommand{\TDCGS}{TDCGS\xspace}
\newcommand{\DTS}{DTS\xspace}
\newcommand{\ADTS}{ADTS\xspace}
\newcommand{\CTS}{CTS\xspace}
\newcommand{\ACTS}{ACTS\xspace}
\newcommand{\Agents}{\mathcal{A}}
\newcommand{\A}{\ensuremath{A}}
\newcommand{\Opp}{\ensuremath{O}} 
\newcommand{\Locations}{\ensuremath{L}}
\newcommand{\loc}{\ensuremath{l}}
\newcommand{\Events}{\ensuremath{Act}}    
\newcommand{\JEvents}{\ensuremath{\mathit{JAct}}}

\newcommand{\evt}{\ensuremath{a}}    
\newcommand{\jevt}{\ensuremath{\alpha}}
\newcommand{\Prot}{\ensuremath{P}}
\newcommand{\Trans}{\ensuremath{T}}

\newcommand{\Agent}{\ensuremath{Agent}}

\newcommand{\model}{\mathit{M}}
\newcommand{\States}{\ensuremath{S}}
\renewcommand{\state}{\ensuremath{s}}
\newcommand{\TStates}{\ensuremath{\mathcal{TS}}}
\newcommand{\tstate}{\ensuremath{q}}
\newcommand{\CStates}{\ensuremath{\mathcal{CS}}}
\newcommand{\cstate}{\ensuremath{q}}
\newcommand{\Clocks}{\mathcal{X}}
\newcommand{\Invariant}{\mathcal{I}}
\newcommand{\Constraints}{\mathcal{C}_\Clocks}
\newcommand{\Constraintsi}{\mathcal{C}_{\Clocks_i}}
\newcommand{\cc}{\mathfrak{cc}}
\newcommand{\PV}{\mathit{PV}}
\newcommand{\prop}[1]{\ensuremath{\mathsf{{#1}}}}
\newcommand{\Val}{\ensuremath{V}}

\def\AG {AG}

\newcommand{\coop}[1]{\langle\!\langle{#1}\rangle\!\rangle}
\newcommand{\Epath}{\mathsf{\exists}}
\newcommand{\Apath}{\mathsf{\forall}}
\newcommand{\Next}{\mathtt{X}\,}
\newcommand{\Sometm}{\mathtt{F}\,}
\newcommand{\Always}{\mathtt{G}\,}
\newcommand{\Until}{\,\mathtt{U}\,}
\newcommand{\Release}{\,\mathtt{R}\,}

\newcommand{\seq}{\pi}
\newcommand{\satisf}[1][]{\models_{_{#1}}}

\newcommand{\strat}{\sigma}
\newcommand{\stratstyle}[1]{\ensuremath{\mathrm{#1}}}
\newcommand{\strattype}{\ensuremath{Y}}
\newcommand{\ir}{\stratstyle{ir}\xspace}
\newcommand{\iR}{\stratstyle{iR}\xspace}
\newcommand{\Ir}{\stratstyle{Ir}\xspace}
\newcommand{\IR}{\stratstyle{IR}\xspace}

\newcommand{\IrT}{\stratstyle{IrT}\xspace}

\newcommand{\IRT}{\stratstyle{IRT}\xspace}

\newcommand{\outcome}{\mathit{out}}
\newcommand{\outm}{\mathcal{E}}
\newcommand{\history}{\mathit{h}}
\newcommand{\History}{\mathit{H}}
\newcommand{\timetype}{\ensuremath{Z}}

\newcommand{\lan}[1]{\ensuremath{\mathbf{#1}}\xspace}
\newcommand{\LTL}[1][]{\lan{LTL_{\stratstyle{#1}}}}
\newcommand{\CTL}[1][]{\lan{CTL}}
\newcommand{\CTLK}[1][]{\lan{CTLK}}
\newcommand{\SCTLK}[1][]{\lan{SCTL_{\stratstyle{#1}}K}}
\newcommand{\CTLs}{\lan{CTL^*}}
\newcommand{\STCTLs}{\lan{STCTL^*}}
\newcommand{\ATL}[1][]{\lan{ATL_{\stratstyle{#1}}}}
\newcommand{\ATLs}[1][]{\lan{ATL_\stratstyle{#1}^*}}

\newcommand{\TCTL}[1][]{\lan{TCTL_{\stratstyle{#1}}}}

\newcommand{\TATL}[1][]{\lan{TATL_{\stratstyle{#1}}}}
\newcommand{\TATLs}[1][]{\lan{TATL_\stratstyle{#1}^*}}

\newcommand{\MTL}[1][]{\lan{MTL}}
\newcommand{\MITL}[1][]{\lan{MITL}}
\newcommand{\TPTL}[1][]{\lan{TPTL}}
\newcommand{\SCTL}[1][]{\lan{SCTL_{\stratstyle{#1}}}}
\newcommand{\STCTL}[1][]{\lan{STCTL_{\stratstyle{#1}}}}

\newcommand{\GL}{\lan{GL}}

\newcommand{\ATLu}[1][]{\lan{ATL_{\stratstyle{#1}}^U}}
\newcommand{\ATLsu}[1][]{\lan{ATL_{\stratstyle{#1}}^{*U}}}
\newcommand{\SCTLu}[1][]{\lan{SCTL_{\stratstyle{#1}}^U}}
\newcommand{\TATLd}[1][]{\lan{TATL_{\stratstyle{#1}}^D}}
\newcommand{\STCTLc}[1][]{\lan{STCTL_{\stratstyle{#1}}^C}}

\newcommand{\TATLsc}[1][]{\lan{TATL_{\stratstyle{#1}}^{*C}}}
\newcommand{\untm}{\stratstyle{U}\xspace}
\newcommand{\disc}{\stratstyle{D}\xspace}
\newcommand{\cont}{\stratstyle{C}\xspace}

\newcommand{\mcheck}[3]{\ensuremath{\textsc{mcheck{#1}}_{#2}({#3})}}

\newcommand{\complexity}[1]{\ensuremath{\mathbf{#1}}}

\newcommand{\V}{\ensuremath{V}\xspace} 
\newcommand{\EA}{\ensuremath{EA}\xspace} 
\newcommand{\imitator}{IMITATOR\xspace}

\newcommand{\lexpr}{\ensuremath{\preceq_e}}
\newcommand{\ldist}{\ensuremath{\preceq_d}}

\setcopyright{none}
\renewcommand\footnotetextcopyrightpermission[1]{}

\title[STCTL]{Strategic (Timed) Computation Tree Logic}

\author{Jaime Arias}
\affiliation{
  \institution{LIPN, CNRS UMR 7030,\\Universit\'{e} Sorbonne Paris Nord}
  \city{Villetaneuse}
  \country{France}}
\email{arias@lipn.univ-paris13.fr}

\author{Wojciech Jamroga}
\affiliation{
  \institution{Institute of Computer Science,\\Polish Academy of Sciences, and}
  \city{SnT, University of Luxembourg}
  \country{}}
\email{jamroga@ipipan.waw.pl}

\author{Wojciech Penczek}
\affiliation{
  \institution{Institute of Computer Science,\\Polish Academy of Sciences}
  \city{Warsaw}
  \country{Poland}}
\email{penczek@ipipan.waw.pl}

\author{Laure Petrucci}
\affiliation{
  \institution{LIPN, CNRS UMR 7030,\\Universit\'{e} Sorbonne Paris Nord}
  \city{Villetaneuse}
  \country{France}}
\email{petrucci@lipn.univ-paris13.fr}

\author{Teofil Sidoruk}
\affiliation{
  \institution{\textsuperscript{1} Institute of Computer Science, PAS}
	\institution{\textsuperscript{2} Faculty of Math. and Inf. Science, Warsaw University of Technology}
  \city{}
  \country{}}
\email{t.sidoruk@ipipan.waw.pl}

\begin{abstract}
We define extensions of \CTL and \TCTL with strategic operators,
called Strategic \CTL (\SCTL) and Strategic \TCTL (\STCTL), respectively.
For each of the above logics we give a synchronous and asynchronous semantics,
\ie{} \STCTL is interpreted over networks of extended Timed Automata (TA)
that either make synchronous moves or synchronise via joint actions.
We consider several semantics regarding information:
imperfect (i) and perfect (I), and
recall: imperfect (r) and perfect (R).
We prove that \SCTL is more expressive than \ATL for all semantics,
and this holds for the timed versions as well.
Moreover,
the model checking problem for \SCTL[\ir] is of the same complexity as for \ATL[\ir],
the model checking problem for \STCTL[\ir] is of the same complexity as for \TCTL,
while for \STCTL[\iR] it is undecidable as for \ATL[\iR].
The above results suggest to use \SCTL[\ir] and \STCTL[\ir] in practical applications.
Therefore, we use the tool \imitator to support model checking of \STCTL[\ir].
\end{abstract}

\keywords{timed automata; model checking; timed logics; strategy logics}

\newcommand{\BibTeX}{\rm B\kern-.05em{\sc i\kern-.025em b}\kern-.08em\TeX}

\begin{document}

\maketitle

\pagestyle{plain}
\pagenumbering{gobble}


\section{Introduction}
\label{sec:intro}

Alternating-time temporal logics \ATLs\ and \ATL~\cite{Alur97ATL,Alur02ATL}
extend the temporal logic \CTLs\ and \CTL, resp., with the notion of \emph{strategic ability}.
These logics allow for expressing properties of agents (or groups of agents) referring
to what they can achieve.
Such properties can be useful for specification, verification, and reasoning about interaction 
in multi-agent systems 
\cite{Kacprzak04umc-atl,KacprzakP05,Lomuscio15mcmas,JamrogaKP16,Huang14symbolic-epist}.

In this paper we investigate timed extensions of strategy logics,
these already known as well as newly introduced ones. 
One of our main aims is to identify the most expressive logics for which the model checking problem 
is not only decidable, but also of complexity acceptable in practice.
We start with recalling the syntax of \ATL \cite{Alur97ATL,Alur02ATL}
and \TATL \cite{TATL2006}.
Then, we put forward definitions of two new logics: 
Strategic CTL (\SCTL), 
and its timed extension, Strategic Timed CTL (\STCTL).
For each (timed) strategy logic we consider two types of interpretations, over 
models of synchronous (Time) Multi-Agent Systems \MAS 
and asynchronous (Time) Multi-Agent Systems \AMAS.
In addition, Time \MAS and Time \AMAS can be either discrete (D), or continuous (C).
We investigate the model checking problem for \SCTL and \STCTL for
all the semantics, and compare their complexity with other strategy logics.
Notably, we prove that \SCTL is more expressive than \ATL for all semantics, 
and this holds for the timed versions as well.
Moreover, 
the model checking problem for \SCTL[\ir] is of the same complexity as for \ATL[\ir],
the model checking problem for \STCTL[\ir] is of the same complexity as for \TCTL, 
while for \STCTL[\iR] it is undecidable as for \ATL[\iR].
These results suggest to use \SCTL[\ir] and \STCTL[\ir] in practical applications.
Therefore, we demonstrate the feasibility of \STCTL[\ir]  model checking
on a small scalable example using \imitator.

\paragraph*{Related Work.}

\TATL\ \cite{TATL2006} is a discrete-time extension of \ATL \cite{Alur97ATL,Alur02ATL}, 
the subset of \ATLs where each strategic modality is immediately followed  by a single temporal operator.
A hierarchy of semantic variants of \TATL\ was established and studied in \cite{KnapikAPJP19}, 
including counting strategies.
Game Logic (\GL)~\cite{Alur02ATL}, similarly to \SCTL, 
combines path quantifiers with the notion of strategic ability.
\GL is a generalisation of \ATLs over perfect information,
where quantification is possible separately over paths within a strategy outcome.
\CTL timed games \cite{Faella2014} are defined over timed automata with continuous time,
but with specifications given using \CTL and \LTL,
placing them somewhere between untimed \SCTLu[\IR] and \STCTLc[\IR] considered here.
They are shown to be \complexity{EXPTIME}-complete,
demonstrating that model checking of \SCTL[\IR] over continuous time models retains the same complexity as over untimed ones (cf.~\Cref{tab:sync:complexity}).
Ana\-logously to \ATLs, the logic \TATLs\ over continuous time semantics, 
call it \TATLsc, would be a natural counterpart to the discrete-time \TATLs.
However, even without the strategic modality, model checking is 
undecidable for continuous time extensions of \LTL (\MTL and \TPTL) \cite{Bouyer09}. 
This has motivated our choice of \STCTL,
which is applicable where discrete time is insufficient,
more expressive than \TATL,
and yet its model checking is decidable for \stratstyle{r}-strategies.

\paragraph*{Outline.}

First, \Cref{sec:logic} recalls the basic notions of strategic logics.
Synchronous systems are tackled in untimed, discrete time and continuous time settings,
all described in an homogeneous manner.
\Cref{sec:semantics} discusses different types of strategies and gives the semantics of considered logics.
Theoretical results regarding the model checking complexity and expressiveness of \lan{S(T)CTL} are introduced
in \Cref{sec:modch} and \Cref{sec:expressivity}, respectively.
\Cref{sec:async} considers asynchronous systems, pointing out differences from the synchronous case wherever applicable.
\Cref{sec:expe} reports experimental results using the IMITATOR model checker.
Finally, \Cref{sec:conclu} concludes the paper.

\section{\mbox{Reasoning about Strategies and Time}}
\label{sec:logic}

In this section, we define the logical framework to reason about strategic abilities in timed synchronous multi-agent systems.
Our definitions are based on~\cite{AlurD90,Lomuscio97interpreted,Alur02ATL,Schobbens04,TATL2006,KnapikAPJP19}.

\subsection{Syntax of \texorpdfstring{\STCTL}{STCTL} and its Fragments}
\label{sec:syntax}

We begin by introducing the logical formulas of interest.
Assume a countable set $\PV$ of atomic propositions, and a finite set $\Agents$ of agents.
The syntax of \emph{Strategic Timed Computation Tree Logic} (\STCTL) resembles that of Game Logic~\cite{Alur02ATL},
and can be defined by the grammar:
\begin{center}
$\varphi ::= \prop{p} \mid \neg \varphi \mid \varphi\wedge\varphi \mid \coop{A}\gamma$,\\
$\gamma ::= \varphi \mid \neg\gamma \mid \gamma \wedge\gamma 
          \mid \Apath\Next\gamma \mid 
          \Apath\gamma\Until_I\gamma \mid \Epath\gamma\Until_I\gamma \mid \Apath\gamma\Release_I\gamma \mid \Epath\gamma\Release_I\gamma$,
\end{center}
where $\prop{p}\in\PV$ is an atomic proposition, $\A\subseteq\Agents$ is a subset of agents, and $I\subseteq\mathbb{R}_{0+}$ is
an interval with bounds of the form $[n,n']$, $[n,n')$, $(n,n']$, $(n,n')$, $(n,\infty)$, $[n,\infty)$, for $n,n'\in\mathbb{N}$.
$\coop{\A}$ is the strategic operator expressing that the agents in $\A$ have a strategy to enforce the temporal property that follows after it.
$\Apath$ (``for all paths") and $\Epath$ (``there exists a path") are the usual path quantifiers of \CTL.
The temporal operators $\Next,\Until,\Release$ stand for ``next'', ``strong until'', and ``release,'' respectively.
Boolean connectives and the remaining operators $\Sometm$ (``eventually''), $\Always$ (``always'') can be derived as usual.
Notice that we added the next step operator $\Next$ to the syntax of \STCTL (and \TCTL) in order to be able to define 
other logics as syntactic fragments of \STCTL, as follows:
\begin{description}
\item[\SCTL:] untimed Strategic \CTL, obtained from \STCTL by restricting the time intervals $I$ in temporal operators
     $\Until_I,\Release_I$ to $I=[0,\infty)$. So, they are removed from the syntax of \STCTL;
\item[\TCTL:] Timed \CTL, obtained from \STCTL by removing the strategic modality $\coop{\A}$ from the syntax;
\item[\CTL:] ``vanilla'' \CTL, obtained from \SCTL by removing $\coop{\A}$;
\item[\TATL:] Timed Alternating-time Temporal Logic, the fragment of \STCTL where each instance of
     $\coop{\A}$ is immediately followed by $\forall\Next$, $\forall\Until_I$, or $\forall\Release_I$, 
		 then $\forall$ is removed from the syntax;
\item[\ATL:] ``vanilla'' \ATL, obtained from \TATL by restricting the time intervals to $I=[0,\infty)$,
             thus removing them from the syntax;
\end{description}
We can now introduce the syntax and semantics of synchronous MAS
with continuous and discrete time, as well as untimed.

\subsection{Continuous Time Synchronous \MAS}
\label{sec:clocks-intro}
\label{sec:sync}

In \emph{synchronous continuous-time multi-agent systems}, all agents have an associated
set of clocks. All clocks evolve at the same rate (across agents), thus allowing for
delays and instantaneous actions.
We first recall the formal definitions for key notions of \emph{timed systems}~\cite{AlurD90}.
Then, we combine them with the concept of \emph{interpreted systems}~\cite{Lomuscio97interpreted}
which has been successful in modelling synchronous MAS.

\emph{Clocks} are non-negative, real-valued variables;
we denote a finite set of clocks by $\Clocks = \{x_1,\dots,x_{n_\Clocks}\}$
(with a fixed ordering assumed for simplicity).
A \emph{clock valuation} on $\Clocks$ is a $n_\Clocks$-tuple $v$.
We denote:
\begin{itemize}
\item by $v(x_i)$ or $v(i)$, the value of clock $x_i$ in $v$;
\item by $v + \delta$, where $\delta \in \Reals_{0+}$, $v'$ s.t. $v'(x) = v(x) + \delta$ for all $x \in \Clocks$;
\item by v[X := 0], where $X \subseteq \Clocks$, $v'$ s.t. $v'(x) = 0$ for all $x \in X$, and $v'(x) = v(x)$ for all $x \in \Clocks \setminus X$.
\end{itemize}
\noindent
The \emph{clock constraints} over $\Clocks$ are defined by the following grammar:
$\cc := true \mid x_i \sim c \mid x_i - x_j \sim c \mid \cc \land \cc$,
where $x_i,x_j \in \Clocks$, $c \in \mathbb{N}$, and $\sim \in \{\leq,<,=,>,\geq\}$.
The set $\Constraints$ collects all constraints over $\Clocks$.
For $\cc \in \Constraints$, the satisfaction relation $\models$ is inductively defined as:
\begin{itemize}
\item[] $v \models true$,
\item[] $v \models (x_i \sim c)$ iff $v(x_i) \sim c$,
\item[] $v \models (x_i - x_j \sim c)$ iff $v(x_i) - v(x_j) \sim c$, and
\item[] $v \models (\cc \land \cc')$ iff $v \models \cc$ and $v \models \cc'$.
\end{itemize}
The set of all valuations satisfying $\cc$ is denoted by $\llbracket\cc\rrbracket$.

\begin{definition}[\TMAS]
\label{def:TMAS:def}
A \emph{continuous-time multi-agent system (\TMAS)} consists of $n$ agents
$\Agents = \set{1,\dots,n}$, each associated with a 9-tuple
$\AG_i = \left(\Locations_i, \iota_i, \Events_i, \Prot_i, \Clocks_i, \Invariant_i, \Trans_i,\PV_i, \Val_i \right)$
including:
\begin{itemize}
\item a finite non-empty set of \emph{local states} $\Locations_i=\{\loc_i^1, \loc_i^2,\dots, \loc_i^{n_i}\}$;
\item an \emph{initial local state} $\iota_i\in \Locations_i$;
\item a finite non-empty set of \emph{local actions} $\Events_i\!=\!\{\evt_i^1,\evt_i^2,\ldots,\evt_i^{m_i}\}$;
\item a \emph{local protocol} $\Prot_i: \Locations_i \to 2^{\Events_i} \setminus \{\emptyset\}$;
\item a set of \emph{clocks} $\Clocks_i$;
\item an \emph{invariant} $\Invariant_i \colon \Locations_i \to \Constraintsi$
	    specifying a condition for the \TMAS to stay in a given local state;
\item a (partial) \emph{local transition function}
      $\Trans_i: \Locations_i \times \JEvents \times \Constraintsi \times 2^{\Clocks_i} \fpart \Locations_i$, where
      $\JEvents \triangleq \prod_{i \in \Agents} \Events_i$ is the set of \emph{joint (global) actions} of all agents,
      is \st $\Trans_i(\loc_i,\jevt,\cc,X) = \loc'_i$ for some $\loc'_i \in \Locations_i$ iff $\evt^i \in \Prot_i(\loc_i)$,
      $\cc \in \Constraintsi$, and $X \subseteq \Clocks_i$;
\item a finite non-empty set of \emph{local propositions} $\PV_i=\{\prop{p_i^1},\ldots,\prop{p_i^{r_i}}\}$;
\item a \emph{local valuation function} $\Val_i: \Locations_i \rightarrow 2^{\PV_i}$.
\end{itemize}
\end{definition}
For a local transition $t := \loc \xrightarrow[]{\alpha,\cc,X} \loc'$ in a \TMAS,
$\loc$ and $\loc'$ are the \emph{source} and \emph{target} states,
$\alpha$ is the executed \emph{action},
clock condition $\cc$ is called a \emph{guard},
and $X$ is the set of clocks to be \emph{reset}.

\begin{definition}[Model of \TMAS]
\label{def:TMAS:model}
The \emph{model} of \TMAS\ is a 7-tuple $\model=(\Agents,\States,\iota,\Clocks,\Invariant,\Trans,\Val)$, where:
\begin{itemize}
\item $\Agents = \set{1,\dots,n}$ is the set of \emph{agents};
\item $\States = \prod_{i=1}^n \Locations_i$ is the set of \emph{global states};
\item $\iota=(\iota_1,\dots,\iota_n) \in \States$ is the \emph{initial global state};
\item $\Clocks = \bigcup_{i \in \Agents}\Clocks_i$ is the set of \emph{clocks};
\item $\Invariant(\state) = \bigwedge_{i \in \Agents} \Invariant_i(\state^i)$ is the
      \emph{global invariant},
			where	$\state^i$ is the $i$-th local state of $\state$;
\item $\Trans: \States \times \JEvents \times \Constraints \times 2^{\Clocks} \rightarrow \States$   
	    \st
	$\Trans(\state,\jevt,\bigwedge_{i} \cc_i,\bigcup_{i=1}^n X_i) = \state'$ iff
	$\Trans_i(\state^i,\jevt,\cc_i,X_i) = \state'^i$ for each $1 \leq i \leq n$;
\item a \emph{valuation function} $\Val\colon\States\to2^\PV$, where $\PV = \bigcup_{i=1}^n \PV_i$.
\end{itemize}
\end{definition}
The continuous (dense) semantics of time defines
\emph{concrete states} as tuples of global states and non-negative real clock valuations.

\begin{definition}[\CTS]
\label{def:TMAS:semantics}
The \emph{concrete model} of a \TMAS model $\model = (\Agents,\States,\iota,\Clocks,\Invariant,\Trans,\Val)$
is given by its \emph{Continuous Transition System} (\CTS) $(\Agents,\CStates,\cstate_\iota,\to_c,V_c)$, where:
\begin{itemize}
\item $\Agents = \set{1,\dots,n}$ is the set of \emph{agents};
\item $\CStates = \States \times \Reals_{0+}^{n_\Clocks}$ is the set of \emph{concrete states};
\item $\cstate_\iota = (\iota,v) \in \CStates$, such that $v(x_i)=0$ for each $x_i \in \Clocks$, is the \emph{concrete initial state};
\item $\to_c \subseteq \CStates \times (\JEvents \cup \Reals_{0+}) \times \CStates$ is the \emph{transition relation},
			defined by time- and action successors as follows:\\
$(\state,v) \xrightarrow[]{\delta}_c (\state,v+\delta)$
				for $\delta \in \Reals_{0+}$ and $v, v+\delta \in \llbracket \Invariant(\state) \rrbracket$,\\
$(\state,v) \xrightarrow[]{\jevt}_c (\state',v')$
				iff there are $\jevt \in \JEvents$, $\cc \in \Constraints$, $X \subseteq \Clocks$~s.t.:
				$\state \xrightarrow[]{\jevt,\cc,X} \state' \in \Trans$,
				$v \in \llbracket\cc\rrbracket$,
				$v \in \llbracket\Invariant(\state)\rrbracket$,
				$v' = v[X := 0]$,
				$v'~\in~\llbracket\Invariant(\state')\rrbracket$,
\item $V_c(\state,v) = \Val(\state)$ is the \emph{valuation function}.
\end{itemize}
\end{definition}
Intuitively, there are two types of transitions:
delay transitions $\xrightarrow[]{\delta}_c$,
which increase the clock valuation(s) by a given $\delta$ but do not change the global state,
and action transitions $\xrightarrow[]{\jevt}_c$ which correspond to executing
an enabled action in the \TMAS and move the latter to a successor state,
possibly resetting some clocks.
Note that if the set of clocks is empty, the concrete model contains only action transitions, 
and thus it is identical with the model itself.

\subsection{Discrete Time Synchronous \MAS}

Synchronous models with discrete time were considered for reasoning 
in \TATL in~\cite{TATL2006,KnapikAPJP19} using Tight Durational Concurrent Game Structures (\TDCGS),
which is a flat model, as opposed to a network of synchronising models, used in this paper.
This gives an equivalent model whose definition is consistent with \MAS. 
Indeed, \emph{synchronous discrete time} \MAS
extend \MAS with a constant duration associated with each individual transition.

\begin{definition}[\DMAS]
\label{def:DMAS:def}
A \emph{discrete-time multi-agent system (\DMAS)} consists of $n$ agents
$\Agents = \set{1,\dots,n}$, where each agent $i \in \Agents$ is associated with a 7-tuple
$\AG_i = \left( \Locations_i,\iota_i, \Events_i, \Prot_i, \Trans_i,\PV_i, \Val_i \right)$ including:
\begin{itemize}
\item a finite non-empty set of \emph{local states} $\Locations_i=\{\loc_i^1, \loc_i^2,\dots, \loc_i^{n_i}\}$;
\item an \emph{initial local state} $\iota_i\in \Locations_i$;
\item a finite non-empty set of \emph{local actions} $\Events_i\!=\!\{\evt_i^1,\evt_i^2,\ldots,\evt_i^{m_i}\}$;
\item a \emph{local protocol} $\Prot_i: \Locations_i \to 2^{\Events_i} \setminus \{\emptyset\}$
      selecting the actions available at each local state;
\item a (partial) \emph{local transition function} $\Trans_i: \Locations_i \times \JEvents
			\rightarrow \Locations_i \times \mathbb{N}_{+}$ such that $\Trans_i(\loc_i,\jevt)$ is defined iff
			$\alpha^i \in \Prot_i(\loc_i)$, where
			$\jevt^i$ is the action of agent $i \in \Agents$ in the \emph{joint action} $\jevt \in \JEvents$;
			$\Trans_i(\loc_i,\jevt)=(\loc'_i,\delta_i)$ where $\delta_i$ is the \emph{duration} of the transition;
\item a finite non-empty set of \emph{local propositions} $\PV_i=\{\prop{p_i^1},\ldots,\prop{p_i^{r_i}}\}$;
\item a \emph{local valuation function} $\Val_i: \Locations_i \rightarrow 2^{\PV_i}$.
\end{itemize}
\end{definition}

This definition is similar to \cref{def:MAS:def} with a duration associated with each
transition.
As for \TMAS, the model of a \DMAS is defined, that describes its behaviour.

\begin{definition}[Model of \DMAS]
\label{def:DMAS:model}
The \emph{model} of \DMAS\ is a 5-tuple $\model=(\Agents,\States,\iota,\Trans,\Val)$, where:
\begin{itemize}
\item $\Agents = \set{1,\dots,n}$ is the set of \emph{agents};
\item $\States = \prod_{i=1}^n \Locations_i$ is the set of \emph{global states};
\item $\iota=(\iota_1,\dots,\iota_n) \in \States$ is the \emph{initial global state};
\item $\Trans: \States\times \JEvents \rightarrow \States \times\mathbb{N}_{+}$ is the partial \emph{global
			transition function}, such that $\Trans(\state,\jevt)= (\state',\delta)$ iff $\Trans_i(\state^i,\jevt) =
			(\state'^i,\delta)$ for all $i \in \Agents$,
			where $s^i$ is the $i$-th local state of $s$;
\item $\Val: \States \rightarrow 2^{\PV}$ is the \emph{valuation function} such that
			$\Val((\loc_1,\ldots,\loc_n)) = \bigcup_{i=1}^n \Val_i(\loc_i)$, where $\PV = \bigcup_{i=1}^n \PV_i$.
\end{itemize}
\end{definition}
This definition enforces all synchronising actions to have the same
duration in their respective local components.
Another possibility could be to have individual durations $\delta_i$ for each component,
and the longest duration $\max_{i=1}^n\delta_i$ for the synchronised transition.
This would mimic actions of the different components
taking place together, with the longest slowing down the whole execution.

\begin{definition}[\DTS]
\label{def:DMAS:semantics}
The concrete model of a \DMAS is given by its \emph{Duration Transition System}
(\DTS) $(\Agents,\TStates,\tstate_\iota,\outm,V_d)$, where:
\begin{itemize}
\item $\Agents = \set{1,\dots,n}$ is the set of \emph{agents};
\item $\TStates = \States\times\mathbb{N}$ is a set of \emph{timed states};
\item $\tstate_\iota=(\iota,0)\in\TStates$ is the \emph{initial timed state};
\item $\outm\colon\TStates\times\JEvents\to\TStates$ is a (partial)
			\emph{transition function} such that
			$\outm((\state, d), \jevt) = (\state', d+\delta) \text{ iff }
			\Trans(\state, \jevt) = (\state', \delta)$,
			for $\state,\state'\in \States$, $\jevt\in \JEvents$,
			$d \in\mathbb{N}$, and $\delta \in\mathbb{N}_{+}$;
\item $V_d(\state,d) = \Val(\state)$ is the \emph{valuation function}.
\end{itemize}
\end{definition}
It is straightforward to see that these compositional definitions correspond to those of
the flat structures (\TDCGS and \DTS) in~\cite{KnapikAPJP19}.

\subsection{Untimed \MAS}

Untimed \MAS can be defined as Timed \MAS with no clocks, see below.
Note that the definition is essentially equivalent to the concept of
an \emph{interpreted system} in~\cite{Lomuscio97interpreted}.
\begin{definition}[\MAS]\label{def:MAS:def}
An \emph{untimed multi-agent system}, simply \MAS, is a \TMAS with every $\Clocks_i = \emptyset$.
The model and concrete model of \MAS are equal and defined as in Definition~\ref{def:TMAS:model},
without clocks.
\end{definition}

\section{Semantics of Logics}
\label{sec:semantics}

We start with defining strategies and their outcomes.

\subsection{Strategies}
\label{sec:strategies}

The taxonomy proposed by Schobbens~\cite{Schobbens04}
defines four strategy types based on agents' \emph{state information}:
perfect~(\stratstyle{I}) vs.~imperfect~(\stratstyle{i}),
and their \emph{recall of state history}:
perfect~(\stratstyle{R}) vs.~imperfect~(\stratstyle{r}).

Intuitively, a strategy can be seen as a conditional plan that dictates the choice of an agent in each possible situation.
In \emph{perfect information} strategies, agents have complete knowledge about global states of the model and thus can make 
different choices in each one.
Under \emph{imperfect information},
decisions can only be made based on local states.
\emph{Perfect recall} assumes that agents have access to a full history of previously visited states,
whereas under \emph{imperfect recall} only the current state is explicitly known.
Formally:

\begin{itemize}
\item A \emph{memoryless imperfect information $(\ir)$ strategy} for $i\in\Agents$ is a function
			$\strat_i \colon \Locations_i \to \Events_i$ such that $\strat_i(\loc) \in \Prot_i(\loc)$
			for each $\loc \in \Locations_i$.
\item A \emph{memoryless perfect information $(\Ir)$ strategy} for $i\in\Agents$ is a function
			$\strat_i \colon \States \to \Events_i$ such that $\strat_i(\state) \in \Prot_i(\state^i)$
			for each $\state \in \States$.
\item A \emph{perfect recall, imperfect information $(\iR)$ strategy} for $i\in\Agents$ is a function
			$\strat_i \colon \Locations_i^+ \to \Events_i$ s.t. $\strat_i(\history) \in \Prot_i(last(\history))$
			for each $\loc \in \Locations_i$.
\item A \emph{perfect recall, perfect information $(\IR)$ strategy} for $i\in\Agents$ is a function
			$\strat_i \colon \States^+ \to \Events_i$ s.t. $\strat_i(\History) \in \Prot_i(last(\History)^i)$
			for each $\state \in \States$.
\end{itemize}
By $\History \in \States^+$ (resp. $\history \in \Locations_i^+$),
we denote a history of global (resp. $i$'s local) states,
and $last(H)$ (resp. $last(h)$) refers to its last state.
The notion of a strategy can be generalized to an agent coalition $\A\subseteq\Agents$,
whose \emph{joint strategy} $\strat_\A$ is a tuple of strategies, one for each $i\in\A$.

We now formally define executions in concrete models of MAS.

\begin{definition}[Execution]
Let $(\Agents,\CStates,\cstate_\iota,\to_c~,V_c)$ be a \CTS. Its execution from 
$\cstate_0 = (\state_0,v_0)$ is
$\seq = \cstate_0,\delta_0,\cstate'_0,\jevt_0,\cstate_1,\delta_1,\cstate'_1,\jevt_1,\dots$,
where $\cstate_k = (\state_{2k}, v_{2k})$, $\cstate'_k = (\state_{2k+1}, v_{2k+1})$,
such that 
$\;\delta_{k} \in \Reals_{0+}$,
$\jevt_{k} \in \JEvents$,
$\cstate_{k} \xrightarrow[]{\delta_k}_c \cstate'_{k} \xrightarrow[]{\jevt_{k}}_c \cstate_{k+1}$,
for each $k \geq 0$.

An execution of a \DTS $(\Agents,\TStates,\tstate_\iota,\outm,V_d)$ from $\tstate_0$ is
$\seq = \tstate_0,\delta_0,\jevt_0,$ $\tstate_1,\delta_1,\jevt_1,\dots$,
where $\tstate_k = (\state_k,d_k)$, $d_0 = 0$,
s.t.
$\jevt_k \in \JEvents$,
$\delta_k \in \mathbb{N}_{+}$,
$\outm((\state_k,d_k),\jevt_k) = (\state_{k+1},d_{k+1})$,
$d_{k+1} = d_k + \delta_k$,
for each $k \geq 0$.

An execution of a (concrete) model of an untimed \MAS from $\state_0$
is $\seq = \state_0,\jevt_0,\state_1,\jevt_1,\dots$,
s.t.
$\jevt_k \in \JEvents$,
$\state_k \xrightarrow[]{\jevt_k}_c \state_{k+1}$,
for each $k \geq 0$.
\end{definition}
\noindent
Note that if the set of clocks is empty in \TMAS, then the executions of \CTS contain only action transitions,
as in an untimed \MAS.

The outcome of a strategy $\strat_\A$ represents executions where the agents in $\A$ adopt $\strat_\A$,
\ie it is the set of all paths in the model that may occur when the coalition strictly follows the strategy,
while opponents freely choose from \event/s permitted by their protocols.

\begin{definition}[Outcome]\label{def:outcome}
Let $\A\subseteq\Agents$, $\strattype \in \set{\ir,\Ir,\iR,\IR}$,
$\model^\cont$ (resp. $\model^\disc$, $\model^\untm$) be the model of a \TMAS (resp. a \DMAS, an untimed \MAS),
and let $\seq^\cont = \cstate_0,\delta_0,\cstate'_0,\jevt_0,\dots$
(resp. $\seq^\disc = \tstate_0,\delta_0,\jevt_0,\dots$;~$\seq^\untm = \state_0,\jevt_0,\dots$)
be an execution of the corresponding concrete model.

The \emph{outcome} of $\strattype$-strategy $\strat_\A$ in state $g^\timetype$ of the concrete model of $\model^\timetype$,
where $\timetype \in \set{\cont, \disc, \untm}$, $g^\timetype = \cstate_0$ for $\timetype \in \set{\cont, \disc}$, 
and $g^\timetype = \state_0$ for $\timetype = \untm$,
is the set $\outcome_{\model^\timetype}^\strattype(g^\timetype,\strat_\A)$,
such that $\seq^\timetype \in \outcome_{\model^\timetype}^\strattype(g^\timetype,\strat_\A)$
iff for each $m \geq 0$ and each agent $i\in\A$:\\
($\strattype = \ir$): \quad $\jevt_m^i = \strat_i(\state_m^i)$,\vspace{-2mm}\\
($\strattype = \iR$): \quad $\jevt_m^i = \strat_i(\history^\timetype)$, \quad $\history^\timetype = 
																				\begin{cases}
																						\state_0^i,\state_0^i,\state_1^i,\state_1^i,\dots,\state_m^i & \text{if } \timetype = \cont\\
																						\state_0^i,\state_1^i,\dots,\state_m^i & \text{otherwise}
																				\end{cases}$\vspace{-2mm}\\
($\strattype = \Ir$): \quad $\jevt_m^i = \strat_i(\state_m)$,\vspace{-2mm}\\
($\strattype = \IR$): \quad $\jevt_m^i = \strat_i(\History^\timetype)$, \quad $\History^\timetype =
																				\begin{cases}
																						\state_0,\state_0,\state_1,\state_1,\dots,\state_m & \text{if } \timetype = \cont\\
																						\state_0,\state_1,\dots,\state_m & \text{otherwise}
																				\end{cases}$\\
where $\state_j \in \States$ is the global state component
of $\cstate_j$ in $\seq^\cont$ and $\seq^\disc$.
\end{definition}

\subsection{Semantics of \texorpdfstring{\TATL}{TATL} and \texorpdfstring{\STCTL}{STCTL}}

We give the discrete-time semantics of \TATL\ \cite{TATL2006} 
and the continuous-time semantics of \STCTL
for strategies of type $\strattype \in \set{\ir,\Ir,\iR,\IR}$.

\begin{definition}[\TATL Semantics]
\label{def:tatlsem}

Let $\model=(\States,\iota,\Trans,\Val)$ be a \DMAS model,
$(\Agents,\TStates,\tstate_\iota,\outm,V_d)$ be its \DTS,
$s = s_0 \in\States$ a state, 
$\varphi,\psi \in \TATL$,
$\pi = (s_0,d_0),\delta_0,\jevt_0,\dots$ an execution of the \DTS,
and $A \subseteq \Agents$ a set of agents.
The $\strattype$-semantics of \TATL is given by the clauses:

\begin{itemize}
\item $\model, s \models \prop{p}$ iff $\prop{p} \in \Val(s)$,
\item $\model, s \models \neg\varphi$ iff $\model, s \not\models \varphi$,
\item $\model, s \models \varphi\land\psi$ iff $\model, s\models\varphi$ and $\model, s \models\psi$,
\item $\model, s \models \coop{A} X\varphi$ iff there exists a joint $\strattype$-strategy $\strat_A$
s.t. for each $\pi \in \outcome_\model^\strattype(s,\strat_A)$ 
we have $\model, s_1 \models \varphi$,
\item $\model, s \models \coop{A} \varphi U_{I}\psi$ iff there exists a joint $\strattype$-strategy $\strat_A$
s.t. for each $\pi \in \outcome_\model^\strattype(s,\strat_A)$ there exists $i\in \mathbb{N}$
s.t. $d_i \in I$ and $\model, s_i \models \psi$ and
for all $0\le j < i : \model, s_j \models\varphi$,
\item $\model, s \models \coop{A} \varphi R_{I} \psi$ iff there exists a joint
$\strattype$-strategy $\strat_A$ 
s.t. for each $\pi \in \outcome_\model^\strattype(s,\strat_A)$ 
for each $i\in \mathbb{N}$ such that $d_i \in I$ we have $\model, s_i\models\psi$ or
there exists $0\le j < i : \model, s_j\models\varphi$.
\end{itemize}
\end{definition}

\begin{definition}[\STCTL Semantics]
\label{def:stctlsem}
Let $\model = (\Agents,\CStates,\cstate_\iota,\to_c,V_c)$ be a \CTS,
$\cstate_0 = (s,v) = (s_0,v_0) \in \CStates$ a concrete state,
$A \subseteq \Agents$, 
$\varphi,\psi \in \STCTL$,
$\pi = \cstate_0,\delta_0,\cstate'_0,\jevt_0,\ldots$ an execution of $\model$ 
where $\cstate_k = (\state_{2k}, v_{2k})$, $\cstate'_k = (\state_{2k+1}, v_{2k+1})$,
and let $time_{\pi}(\cstate_k)$ = $\sum_{j=0}^{k-1} \delta_j$, $time_{\pi}(\cstate'_k)$ = $\sum_{j=0}^{k} \delta_j$.
The $\strattype$-semantics of \STCTL is given as:
\begin{itemize}
\item $\model, (\state,v) \models \prop{p}$ iff $\prop{p} \in \V_c(\state,v)$,
\item $\model, (\state,v) \models \neg\varphi$ iff $\model, (\state,v) \not\models \varphi$,
\item $\model, (\state,v) \models \varphi\land\psi$ iff $\model, (\state,v)\models\varphi$ and $\model, (\state,v) \models\psi$,
\item $\model, (\state,v) \models \coop{A} \gamma$ iff there exists a joint $\strattype$-strategy $\strat_A$
			such that we have $\model, \outcome_\model^\strattype((\state,v),\strat_A) \models \gamma$, where:
			\begin{itemize}
			\item $\model, \outcome_\model^\strattype((\state,v),\strat_A) \models \prop{p}$ iff $\prop{p} \in \V_c(\state,v)$,
			\item $\model, \outcome_\model^\strattype((\state,v),\strat_A) \models \neg\varphi$ iff
						$\model, \outcome_\model^\strattype((\state,v),\strat_A) \not\models \varphi$,
			\item $\model, \outcome_\model^\strattype((\state,v),\strat_A) \models \varphi \land \psi$ iff
						$\model, \outcome_\model^\strattype((\state,v),\strat_A) \models \varphi$\\ and $\model, \outcome_\model^\strattype((\state,v),\strat_A) \models \psi$,
			\item $\model, \outcome_\model^\strattype((\state,v),\strat_A) \models \forall \Next\varphi$ iff
						for each $\pi \in \outcome_\model^\strattype((\state,v),\strat_A)$ we~have $\model,(\pi,\strat_A,Y) \models \Next\varphi$,
			\item $\model, \outcome_\model^\strattype((\state,v),\strat_A) \models \forall \varphi U_{I}\psi$ iff
						for each $\pi \in \outcome_\model^\strattype((\state,v),\strat_A)$ we have $\model,(\pi,\strat_A,Y) \models \varphi U_{I}\psi$,
			\item $\model, \outcome_\model^\strattype((\state,v),\strat_A) \models \exists \varphi U_{I}\psi$ iff
						for some $\pi \in \outcome_\model^\strattype((\state,v),\strat_A)$ we have $\model,(\pi,\strat_A,Y) \models \varphi U_{I}\psi$,
			\item $\model, \outcome_\model^\strattype((\state,v),\strat_A) \models \forall \varphi R_{I} \psi$ iff
						for each $\pi \in \outcome_\model^\strattype((\state,v),\strat_A)$ we have $\model, (\pi,\strat_A,Y) \models \varphi R_{I}\psi$,
			\item $\model, \outcome_\model^\strattype((\state,v),\strat_A) \models \exists \varphi R_{I} \psi$ iff
						for some $\pi \in \outcome_\model^\strattype((\state,v),\strat_A)$ we have $\model, (\pi,\strat_A,Y) \models \varphi R_{I}\psi$, where
						\begin{itemize}
						\item $\model,(\pi,\strat_A,Y) \models \Next\varphi$ iff 
									$\model, \outcome_\model^\strattype(\state_1,\strat_A) \models \varphi$ \;\; (untimed only)
						\item $\model,(\pi,\strat_A,Y) \models \varphi U_{I}\psi$ iff there is $i\in \mathbb{N}$
									s.t. $time_{\pi}(\state_i,v_i) \in I$ and we have
									$\model, \outcome_\model^\strattype((\state_i,v_i),\strat_A) \models \psi$ and
									for all $0\le j < i : \model, \outcome_\model^\strattype((\state_j,v_j),\strat_A) \models\varphi$,
						\item $\model,(\pi,\strat_A,Y) \models \varphi R_{I}\psi$ iff
									for each $i\in \mathbb{N}$ s.t. $time_{\pi}(\state_i,v_i) \in I$
									we have $\model,\outcome_\model^\strattype((\state_i,v_i),\strat_A) \models\psi$ or
									there exists $0\le j < i : \model, \outcome_\model^\strattype((\state_j,v_j),\strat_A) \models\varphi$.
						\end{itemize}
			\item $\model, \outcome_\model^\strattype((\state,v),\strat_A) \models \coop{A'} \gamma$ iff $\model, (\state,v) \models \coop{A'} \gamma$.
			\end{itemize}
\end{itemize}
\end{definition}		

\noindent
By $\mathcal{L}_\mathcal{S}^\mathcal{M}$, we denote the logical system being considered,
where $\mathcal{L}$ is the syntactic variant (see Section~\ref{sec:syntax}),
$\mathcal{S}$ is the class of strategies (cf.~Section~\ref{sec:strategies}),
and $\mathcal{M}\in\set{\cont,\disc,\untm}$ is the class of continuous-time, discrete-time, and untimed models, respectively.
Superscripts \cont and \untm, if omitted, are assumed to follow from the syntactic variant.

\section{Model Checking Results}
\label{sec:modch}

We now recall complexity results under different semantics for \ATL (\Cref{sec:sync:modch:atl})
and \TATL (\Cref{sec:sync:modch:tatl}).
Then, in \Cref{sec:sync:modch:sctl:r,sec:sync:modch:sctl:R,sec:sync:modch:stctl:r,sec:sync:modch:stctl:R},
we provide new results for \SCTL and \STCTL.
It is important to note that complexity results are given wrt. the \emph{model} size,
as defined in \Cref{sec:logic}.
In particular, note that the model and the concrete model are not equal in each case.
The results are summarised in \Cref{tab:sync:complexity}.

\subsection{Model Checking \texorpdfstring{\ATL}{ATL}}
\label{sec:sync:modch:atl}

The standard fixpoint algorithm for model checking \ATL under perfect information was presented in the original
paper by Alur, Henzinger and Kupferman \cite{Alur02ATL}.
In a nutshell, to verify a formula $\coop{\A}\varphi$, it starts with a candidate set of states
(chosen appropriately depending on $\varphi$) and then iterates backwards
over the abilities of coalition $\A$ at each step \cite{Jamroga15specificationMAS}.
Model checking of \ATL is \complexity{PTIME}-complete \cite{Alur02ATL},
for both memoryless and perfect recall strategies,
since the satisfaction semantics for \ATL[\Ir] and \ATL[\IR] coincide \cite{Jamroga15specificationMAS}.

On the other hand, the fixpoint-based approach cannot be adapted to imperfect information
\cite{Agotnes15handbook,Goranko15stratmas}, making model checking significantly more complex
in this setting:
\complexity{\Delta^P_2} for \ATL[ir] and NP-hard for simple instances
of \ATL[ir] \cite{Jamroga15specificationMAS},
and undecidable for \ATL[iR] \cite{Dima11undecidable}.

\subsection{Model Checking \texorpdfstring{\TATL}{TATL}}
\label{sec:sync:modch:tatl}

Algorithms and complexity results are given in \cite{TATL2006} for model checking \TATL over \TDCGS,
which are analogous to \DMAS.
The strategies considered are defined on histories of concrete states, \ie of type \IRT,
where $T$ refers  to the timed strategy.
Furthermore, analogously to \ATLs, the semantics of \TATL[\IrT] and \TATL[\IRT] coincide \cite{KnapikAPJP19}.
Model checking \TATL[\IrT] and \TATL[\IRT] is \complexity{EXPTIME}-complete in the general case \cite[Theorem 13]{TATL2006},
and \complexity{PTIME}-complete for the subset that excludes equality in time constraints \cite[Theorem 14]{TATL2006}.

It is important to note, though, that these results are given wrt. the formula size;
in particular, the exponential blowup in the general case is attributed
solely to the binary encoding of constraints in \TATL formulas,
while the algorithm is actually in \complexity{PTIME} wrt. the model size \cite[Theorem 12]{TATL2006}.
Hence, we put the latter in \Cref{tab:sync:complexity},
as it cannot be otherwise compared with our results for \lan{S(T)CTL}.

We are not aware of any works investigating \TATL under imperfect information thus far.
However, for \TATL[\ir] we can at least establish upper and lower bounds of \complexity{PSPACE} (since $\TATL \subset \STCTL$)
and \complexity{\Delta^P_2} (since $\ATL \supset \TATL$), respectively.
Furthermore, \TATL[\iR] model checking is undecidable since it is already the case for \ATL[\iR].

\subsection{Model Checking \texorpdfstring{\SCTL}{SCTL} with r-Strategies}
\label{sec:sync:modch:sctl:r}

Under the memoryless semantics of strategic ability,
the complexity of model checking \SCTL can be established analogously to the way it was done for \ATL \cite{Jamroga15specificationMAS}.
We begin with \SCTL[ir].

\begin{algorithm}[htb]
	\caption{\mcheck{\STCTL}{\stratstyle{ir,Ir}}{\model,\state,\coop{\A}\gamma}}
	\label{algo:sctl:ir}
	\begin{itemize}
	\item Guess a strategy $\strat_\A$.
	\item Prune $\model$ by removing transitions not consistent with $\strat_\A$.
	\item If $\coop{\A}\gamma$ is an \STCTL formula,
	      run \TCTL model checking on $\gamma$, else run \CTL model checking on $\gamma$.
	\end{itemize}
\end{algorithm}

\begin{theorem}
\label{theorem:sctl:ir}
Model checking \SCTL[\ir] is \complexity{\Delta^P_2}-complete.
\end{theorem}
\begin{proof}
Let $\varphi = \coop{\A}\gamma$ be an \SCTL[\ir] formula without nested strategic modalities.
Consider the procedure in \Cref{algo:sctl:ir}.
It runs in \complexity{NP} for input $\varphi$,
since $\strat_\A$ can be guessed in non-deterministic polynomial time,
while pruning transitions and model checking of \CTL formula $\gamma$
requires deterministic polynomial time.
Note that for an arbitrary \SCTL[\ir] formula,
nested strategic modalities can be replaced with fresh propositions
by calling \Cref{algo:sctl:ir} recursively, bottom up.
This requires polynomially many calls (wrt. the formula size)
to an \complexity{NP} oracle executing \Cref{algo:sctl:ir},
hence the upper bound of $\complexity{P}^{\complexity{NP}} = \complexity{\Delta^P_2}$.

The lower bound follows from the fact that \SCTL[\ir] subsumes \ATL[\ir] whose model checking complexity
is \complexity{\Delta^P_2} \cite[Th. 34]{Jamroga15specificationMAS}.
\end{proof}

For \SCTL[Ir], i.e., strategies with perfect information, a more involved construction is required to establish the lower bound.

\begin{theorem}
\label{theorem:sctl:Ir}
Model checking \SCTL[\Ir] is \complexity{\Delta^P_2}-complete.
\end{theorem}
\begin{proof}
The upper bound follows exactly as for \SCTL[\ir] (Th.~\ref{theorem:sctl:ir}).

\short{%
The lower bound is obtained by a reduction from the model checking problem for \ATL[\ir],
which is \complexity{\Delta^P_2}-complete~\cite{Jamroga06atlir-eumas}.
We present an overview of the reduction below, and refer the reader to \cite{arXiv-SCTLIr}
for more technical details on the constructions used.

Let $\model$ be an untimed model, and $\varphi$ an \ATL[\ir] formula.
First, we reconstruct the model by cloning states in $M$ so that they record the latest
action profile that has been executed, as in~\cite[Section~3.4]{Goranko04comparingKRA}.
That is, for each state $q$ in $M$ and incoming transition labeled by $(\evt_1,\dots,\evt_n)$,
we create a new state $(q,\evt_1,\dots,\evt_n)$, and direct the transition to that state.
Moreover, we label the new state by fresh atomic propositions
$\prop{exec_{1,\evt_1}},\dots,\prop{exec_{n,\evt_n}}$
that can be used to capture the latest decision of each agent within the formulas of the logic.
We denote the resulting model by $M'$.

Then, we temporarily add epistemic operators $K_i$ to the language, with the standard
observational semantics (\ie{}, $K_i\varphi$ holds in state $s$ iff $\varphi$ is true
in all the states $s'$ such that $(s')^i=s^i$).
The uniformity of agent $i$'s  play can be now captured by the following \CTLK formula:
$\mathit{unif}_i \equiv \Apath\Always(\bigvee_{\evt\in\Events_i}K_i\Apath\Next\prop{exec_{i,\evt}})$.
We reconstruct formula $\varphi$ by replacing every occurrence
of $\coop{A}\psi$ with $\coop{A}(\Apath\psi \land \bigwedge_{i\in A}\mathit{unif}_i)$.
We denote the resulting \SCTLK formula by $\varphi'$.
It is easy to see that $M,s \models_{\ATL[\ir]} \varphi$ iff $M',s \models_{\SCTLK[\Ir]} \varphi'$.

Finally, we do a translation from \SCTLK to \SCTL by a straightforward adaptation
of the construction in~\cite[Section~4.2]{Jamroga08reduction}.
For each agent $i\in\A$, we add an ``epistemic ghost'' $e_i$ to the set of agents.
Then, we simulate the indistinguishability of states in $M'$ with transitions effected by the epistemic ghosts.
That is, we add transitions controlled by $e_i$ between each pair of states $s,s'$ with $s^i=(s')^i$.
We also replace the knowledge operators in $\varphi'$ by appropriate strategic
subformulas for $e_i$, see~\cite[Section~4.2]{Jamroga08reduction} for the details.
The resulting translations of $M'$ and $\varphi'$ are denoted by $M''$ and $\varphi''$.
Analogously to~\cite[Theorem~1]{Jamroga08reduction},
we get that $M',s \models_{\SCTLK[\Ir]} \varphi'$ iff $M'',s \models_{\SCTL[\Ir]} \varphi''$
(note that we need to extend the proof of~\cite[Theorem~1]{Jamroga08reduction}
to Boolean combinations of reachability/safety objectives, but that is also straightforward).
This completes the reduction.
}
\extended{%
The lower bound is obtained by a reduction from the model checking problem for \ATL[\ir],
which is \complexity{\Delta^P_2}-complete~\cite{Jamroga06atlir-eumas}.
Let $\model$ be an untimed model, and $\varphi$ an \ATL[\ir] formula.
The reduction proceeds via the following steps:
\begin{enumerate}
\item Reduction of model checking for \SCTL[\ir] to model checking for \SCTLK[\Ir], i.e.,
\SCTL with memoryless perfect information strategies, and extended with epistemic operators for observational knowledge.

\item Reduction of model checking for \SCTLK[\Ir] to model checking for \SCTL[\Ir]
by a model-specific translation of the epistemic operators to strategic formulas.
\end{enumerate}

\begin{figure}[t]\centering
\includegraphics[width=5cm]{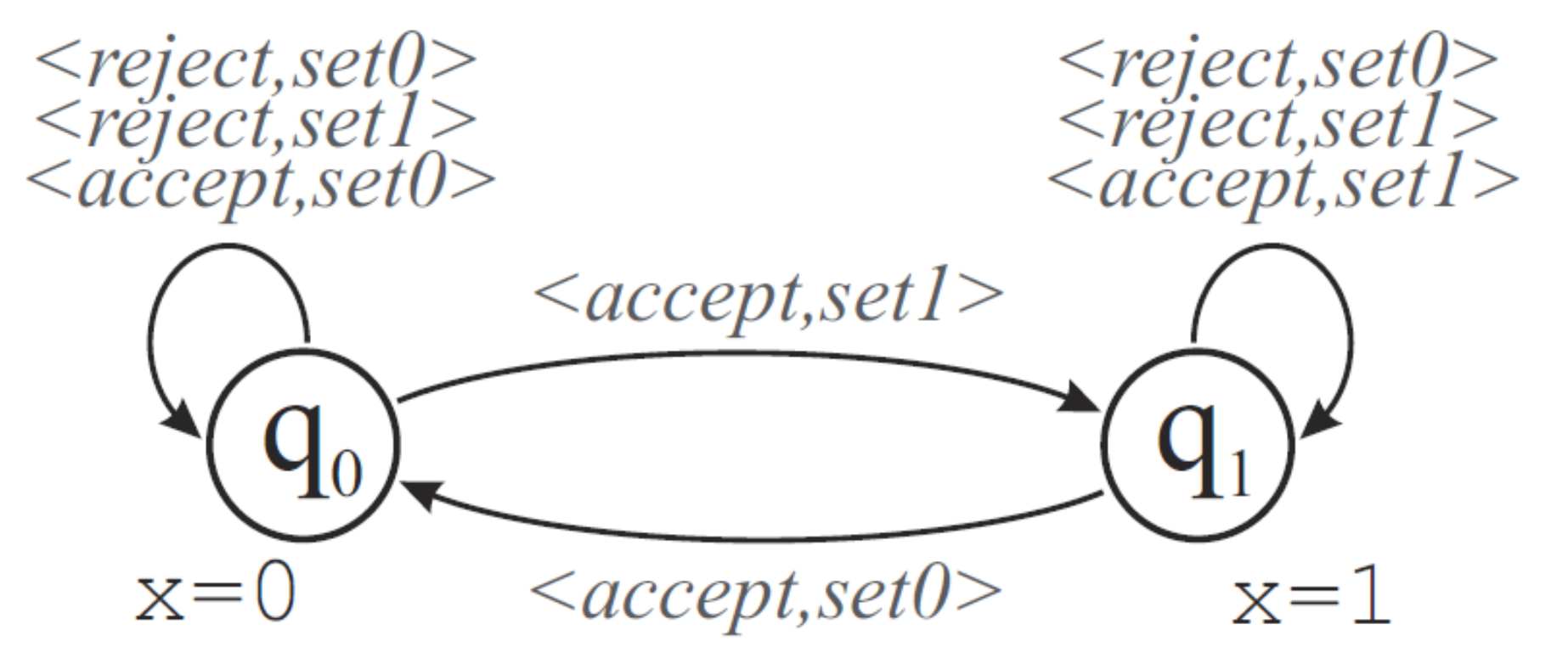}
\caption{A simple model of client/server interaction~\cite{Goranko04comparingKRA}}
\label{fig:cloning-before}
\end{figure}

\para{Model checking reduction from \SCTL[\ir] to \SCTLK[\Ir].}
The idea is to transform the model in such a way that the uniformity of agent $i$'s
choices in the current epistemic neighborhood can be expressed by a temporal-epistemic formula
$\mathit{unif}_i$. Then, the existence of a uniform memoryless strategy for coalition $A$
that achieves $\psi$ is captured by the \SCTLK formula
$\coop{A}(\Apath\psi \land \bigwedge_{i\in A}\mathit{unif}_i)$, with $\coop{A}$ interpreted over \Ir-strategies.
The reduction proceeds as follows.

\begin{figure}[t]\centering
\includegraphics[width=\columnwidth]{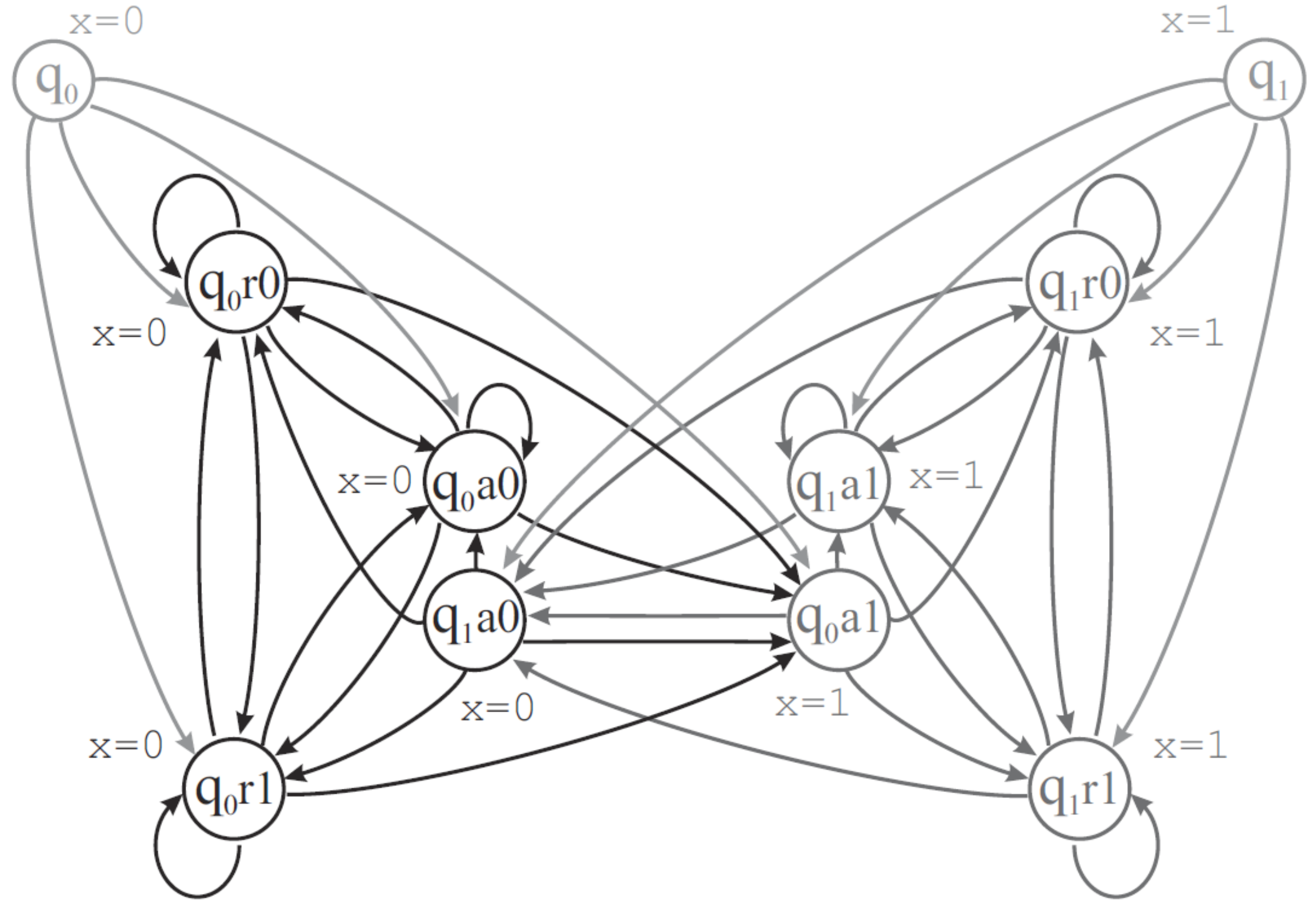}
\caption{The client/server model after cloning~\cite{Goranko04comparingKRA}.
We omit the transition labels and labelling with auxilliary propositions for the sake of readability}
\label{fig:cloning-after}
\end{figure}

\smallskip\noindent
\underline{Model reconstruction}.
First, we reconstruct the model $\model$ by cloning its states so that they record the latest
action profile that has been executed, as in~\cite[Section~3.4]{Goranko04comparingKRA}.
That is, for each state $q$ in $M$ and incoming transition labeled by $(\evt_1,\dots,\evt_n)$,
we create a new state $(q,\evt_1,\dots,\evt_n)$, and direct the transition to that state.
Moreover, we label the new state by fresh atomic propositions
$\prop{exec_{1,\evt_1}},\dots,\prop{exec_{n,\evt_n}}$
that can be used to capture the latest decision of each agent within the formulas of the logic.
We denote the resulting model by $tr_1(\model)$.
An example of the transformation is shown in Figures~\ref{fig:cloning-before} and~\ref{fig:cloning-after}.

\smallskip\noindent
\underline{Capturing uniformity}.
We add epistemic operators $K_i$ to the language of \SCTL, with the standard observational semantics.
That is, $K_i\varphi$ holds in state $s$ iff $\varphi$ holds in all the states $s'$ such that $(s')^i=s^i$.
The uniformity of agent $i$'s  play can be now captured by the following \CTLK formula:
$\mathit{unif}_i \equiv \Apath\Always(\bigvee_{\evt\in\Events_i}K_i\Apath\Next\prop{exec_{i,\evt}})$.
We reconstruct formula $\varphi$ by replacing every occurrence
of $\coop{A}\psi$ with $\coop{A}(\Apath\psi \land \bigwedge_{i\in A}\mathit{unif}_i)$.
We denote the resulting \SCTLK formula by $tr_1(\varphi)$.

\smallskip\noindent
\underline{Correctness of the reduction}.
It is easy to see that \\
\centerline{$\model,\state \satisf[{\ATL[\ir]}] \coop{A}\psi$\ iff\ $\model,\state \satisf[{\SCTLK[\Ir]}] \coop{A}(\Apath\psi \land \bigwedge_{i\in A}\mathit{unif}_i)$.}\\
By straightforward induction on the structure of $\phi$, we get that \\
\centerline{$\model,\state \satisf[{\ATL[\ir]}] \phi$\ if and only if\ $tr_1(\model),\state \satisf[{\SCTLK[\Ir]}] tr_1(\phi)$}\\
for every \ATL formula $\phi$.

\begin{figure}[t]\centering
\includegraphics[width=5cm]{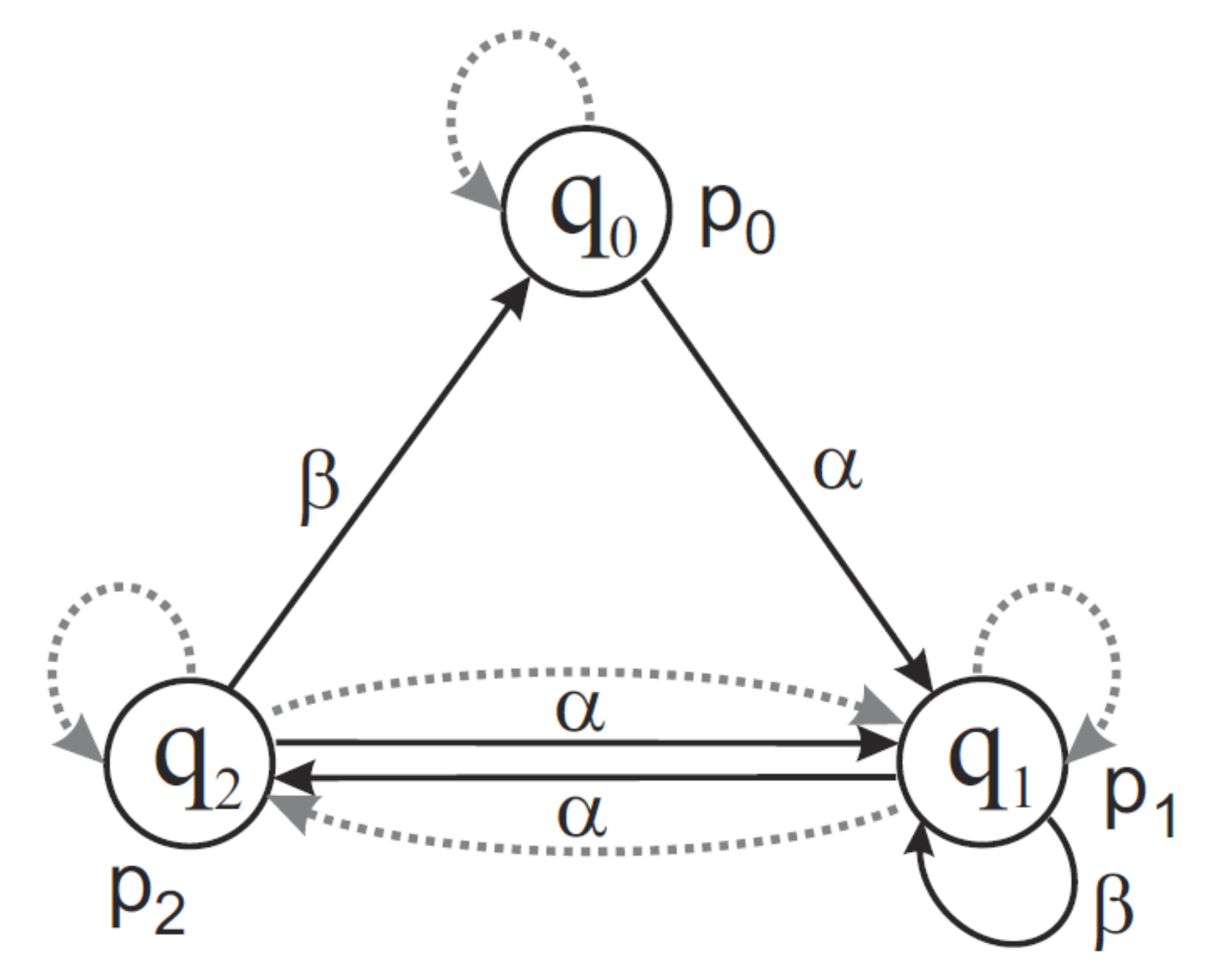}
\caption{Example single-agent game model from~\cite{Jamroga08reduction}}
\label{fig:epistreduct-before}
\end{figure}

\para{Reduction for epistemic operators.}
Finally, we do a translation from \SCTLK to \SCTL by a straightforward adaptation
of the construction proposed in~\cite[Section~3.4]{Goranko04comparingKRA} and refined in~\cite[Section~4.2]{Jamroga08reduction}.
For each agent $i\in\A$, we add an ``epistemic ghost'' $e_i$ to the set of agents.
Then, we simulate the indistinguishability of states in $M'$ with transitions effected by the epistemic ghosts.
That is, we add transitions controlled by $e_i$ between each pair of states $s,s'$ with $s^i=(s')^i$.
We also replace the knowledge operators in $\varphi'$ by appropriate strategic
subformulas for $e_i$, see~\cite[Section~4.2]{Jamroga08reduction} for the details.
The resulting translations of $M'$ and $\varphi'$ are denoted by $M''$ and $\varphi''$.

\begin{figure}[t]\centering
\includegraphics[width=\columnwidth]{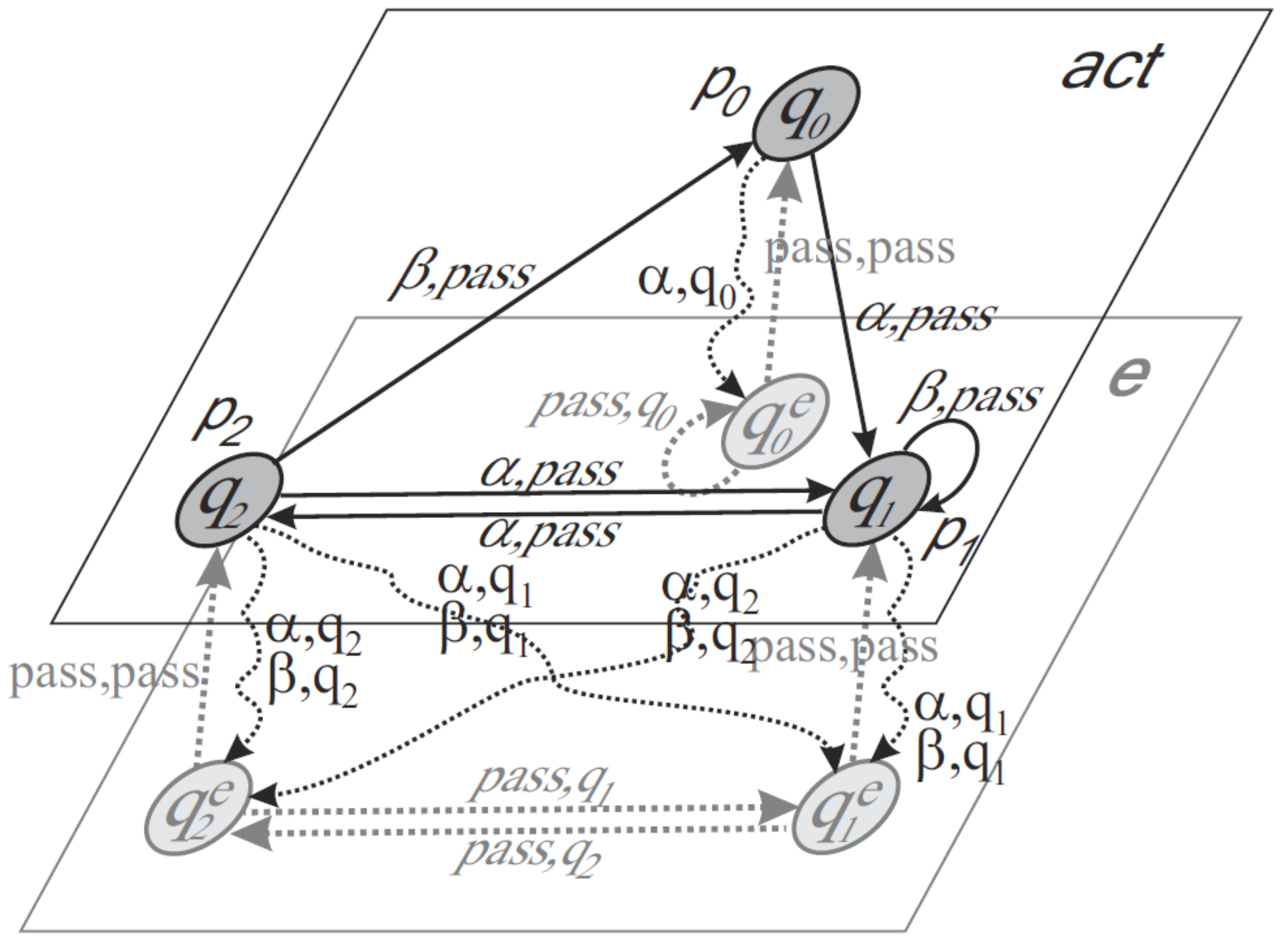}
\caption{Reconstruction of the model in Figure~\ref{fig:epistreduct-before} by adding the epistemic agent $e$ and ``epistemic copies'' of states $q_i^e$}
\label{fig:epistreduct-after}
\end{figure}

Analogously to~\cite[Theorem~1]{Jamroga08reduction},
we get that $M',s \models_{\SCTLK[\Ir]} \varphi'$ iff $M'',s \models_{\SCTL[\Ir]} \varphi''$
(note that we need to extend the proof of~\cite[Theorem~1]{Jamroga08reduction}
to Boolean combinations of reachability/safety objectives, but that is also straightforward).
This completes the reduction.
}
\end{proof}

\subsection{Model Checking \texorpdfstring{\SCTL}{SCTL} with R-Strategies}
\label{sec:sync:modch:sctl:R}

As with the corresponding variants of \ATL and logics that extend it,
we immediately obtain undecidability for \SCTL[\iR].

\begin{theorem}
\label{theorem:sctl:iR}
Model checking \SCTL[\iR] is undecidable.
\end{theorem}
\begin{proof}
Follows from the fact \SCTL subsumes \ATL, whose model checking
is undecidable in the \iR semantics \cite[Theorem 1]{Dima11undecidable}.
\end{proof}

For \SCTL[\IR], we obtain \complexity{EXPTIME}-completeness via a reduction
of the \emph{module checking} problem \cite{KupfermanVW01} for \CTL as follows.

\begin{algorithm}[htb]
	\caption{\mcheck{\SCTL}{\stratstyle{IR}}{\model,\state,\coop{\A}\gamma}}
	\label{algo:sctl:IR}
	\begin{itemize}
	\item Split agents into two groups: coalition $\A$ and opponents $\Opp$.
	\item Merge $\A$ and $\Opp$ into single agents by creating auxiliary
	action labels for tuples of actions belonging to agents in $\A$ and $\Opp$.
	\item Run \CTL module checking on $\gamma$.
	\end{itemize}
\end{algorithm}

\begin{theorem}
\label{theorem:sctl:IR}
Model checking \SCTL[\IR] is \complexity{EXPTIME}-complete.
\end{theorem}
\begin{proof}
(Sketch)
The upper bound follows from the procedure in \Cref{algo:sctl:IR},
which runs in \complexity{EXPTIME},
since \CTL module checking is \complexity{EXPTIME}-complete.
If $\varphi$ has nested coalition operators,
they can be eliminated by proceeding recursively bottom up,
requiring polynomially many calls to \Cref{algo:sctl:IR} (wrt. formula size).
Thus the procedure still runs in \complexity{EXPTIME} for an arbitrary \SCTL[\IR] formula.

The lower bound follows from the fact \CTL module checking,
which is \complexity{EXPTIME}-complete \cite[Th. 3.1]{KupfermanVW01},
can be seen as a special case of \SCTL[\IR] model checking where we have
a single strategic operator at the beginning of a formula and only two agents.
\end{proof}

\subsection{Model Checking \texorpdfstring{\STCTL}{STCTL} with r-Strategies}
\label{sec:sync:modch:stctl:r}

For \STCTL with memoryless strategies, we have that model checking is \complexity{PSPACE}-complete,
\ie it remains unchanged from \TCTL.

\begin{theorem}
\label{theorem:stctl:ir}
Model checking \STCTL[\ir] is \complexity{PSPACE}-complete.
\end{theorem}
\begin{proof}
Let $\varphi = \coop{\A}\gamma$ be an \STCTL[\ir] formula without nested strategic modalities.
The upper bound follows analogously to the case of \SCTL[\ir] (cf.~\Cref{theorem:sctl:ir}).
\Cref{algo:sctl:ir} runs in $\complexity{NPSPACE} = \complexity{PSPACE}$ for input $\varphi$
(since rather than \CTL, \TCTL model checking, which is in \complexity{PSPACE}, is now called on $\gamma$).
For an arbitrary \STCTL[\ir] formula, eliminating nested modalities requires polynomially many calls to \Cref{algo:sctl:ir},
thus we obtain the upper bound of $\complexity{P}^{\complexity{PSPACE}} = \complexity{PSPACE}$.

The lower bound follows from the fact \STCTL subsumes \TCTL, whose model checking
is \complexity{PSPACE}-complete \cite{AlurCD93}.
\end{proof}

\begin{theorem}
\label{theorem:stctl:Ir}
Model checking \STCTL[\Ir] is \complexity{PSPACE}-complete.
\end{theorem}
\begin{proof}
Both bounds follow exactly as in Th.~\ref{theorem:stctl:ir} for \STCTL[\ir].
\end{proof}

\subsection{Model Checking \texorpdfstring{\STCTL}{STCTL} with R-strategies}
\label{sec:sync:modch:stctl:R}

Finally, under perfect recall model checking of \STCTL is undecidable,
which for \iR semantics directly follows from prior results,
and for \IR semantics is obtained via a reduction to \TCTL games \cite{FaellaTM02}.

\begin{theorem}
\label{theorem:stctl:iR}
Model checking \STCTL[\iR] is undecidable.
\end{theorem}
\begin{proof}
Follows from the fact \STCTL subsumes \ATL, whose model checking is undecidable
for \iR strategies \cite[Theorem 1]{Dima11undecidable}.
\end{proof}

\begin{theorem}
\label{theorem:stctl:IR}
Model checking \STCTL[\IR] is undecidable.
\end{theorem}
\begin{proof}
(Sketch)
Undecidability follows from the fact that \TCTL games \cite{FaellaTM02} can be seen as
a special case of \STCTL[IR] model checking,
with a single strategic operator at the beginning of the formula,
and two agents (obtained from grouping together all coalition agents and all opponents as in \Cref{algo:sctl:IR}).
Since TCTL games are undecidable for unrestricted \TCTL \cite[Theorem 3]{FaellaTM02},
clearly this is also the case for the more general case with nested strategic modalities.
\end{proof}

\begin{table}
\renewcommand{\arraystretch}{2}
\begin{tabular}{ |c|c|c|c|c| }
 \hline
 & \Ir & \IR & \ir & \iR \\
 \hline
	\ATLu &
	\multicolumn{2}{c|}{\makecell{\complexity{PTIME} \cite[Th. 5.2]{Alur02ATL}\\(\Ir, \IR semantics coincide)}} &
	\makecell{\complexity{\Delta^P_2} \cite[Th. 34]{Jamroga15specificationMAS}} &
	\multirow{6}{*}{\rotatebox{270}{\complexity{Undecidable} \cite[Th. 1]{Dima11undecidable} ($\supseteq \ATL[\iR]$)}} \\
 \cline{1-4}
	\SCTLu &
	\makecell{\complexity{\Delta^P_2} $[$Th.~\ref{theorem:sctl:Ir}$]$} &
	\makecell{\complexity{EXPTIME}\\$[$Th.~\ref{theorem:sctl:IR}$]$} &
	\makecell{\complexity{\Delta^P_2} $[$Th.~\ref{theorem:sctl:ir}$]$} & \\
 \cline{1-4}
	\ATLsu &
	\makecell{\complexity{PSPACE}\\\cite[Th. 23]{Jamroga15specificationMAS}} &
	\makecell{\complexity{2EXPTIME}\\\cite[Th. 5.6]{Alur02ATL}} &
	\makecell{\complexity{PSPACE}\\\cite[Th. 38]{Jamroga15specificationMAS}} & \\
 \cline{1-4}
	\TATLd &
	\multicolumn{2}{c|}{\makecell{\complexity{PTIME}\tablefootnote{
		\complexity{PTIME} wrt. model size, \complexity{EXPTIME} wrt. formula length (cf.~\Cref{sec:sync:modch:tatl}).}
		$[$\Cref{sec:sync:modch:tatl}$]$\\(\Ir, \IR semantics coincide)}} &
	\makecell{\complexity{\Delta^P_2}---\complexity{PSPACE}} & \\
 \cline{1-4}
	\STCTLc &
	\makecell{\complexity{PSPACE}\\$[$Th.~\ref{theorem:stctl:Ir}$]$} &
	\makecell{\complexity{Undecidable}\\$[$Th.~\ref{theorem:stctl:IR}$]$} &
	\makecell{\complexity{PSPACE}\\$[$Th.~\ref{theorem:stctl:ir}$]$} & \\
 \cline{1-4}
	\TATLsc &
	\multicolumn{3}{c|}{\makecell{\complexity{Undecidable} \cite[Th. 4.3]{Bouyer09} ($\supset \MTL$)}} & \\
 \hline
\end{tabular}
\caption{Model checking complexity wrt. the model size,
for \STCTL, its subsets, and $\ATLsu$, $\TATLsc$ for comparison.
Undecidability (\ATLu[\iR], \MTL) propagates to more expressive logics.}
\label{tab:sync:complexity}
\end{table}

\section{Expressivity Results}
\label{sec:expressivity}

We can see in Section~\ref{sec:modch} that using the broader syntax of \STCTL (\SCTL),
rather than \TATL (\ATL, resp.), does not significantly worsen the complexity of model checking,
especially for the imperfect information semantics.
In this section, we show that, in addition, it strictly increases the expressivity of the logic.
We start by recalling the formal definitions of expressive and distinguishing power.

\begin{definition}[Expressive power and distinguishing power~\cite{Wang09expressive}]
Consider two logical systems ${L}_1$ and ${L}_2$, with their semantics defined
over the same class of models $\mathcal{M}$.
${L}_1$ is \emph{at least as expressive as ${L}_2$} (written ${L}_2 \lexpr {L}_1$) if,
for every formula $\varphi_2$ of ${L}_2$, there exists a formula $\varphi_1$ of ${L}_1$,
such that $\varphi_1$ and $\varphi_2$ are satisfied in the same models from $\mathcal{M}$.

Moreover, ${L}_1$ is \emph{at least as distinguishing as ${L}_2$} (${L}_2 \ldist {L}_1$)
if every pair of models $M,M'\in\mathcal{M}$ that can be
distinguished by a formula of ${L}_2$ can also be distinguished by some formula of ${L}_1$.
\end{definition}
It is easy to see that ${L}_2 \lexpr {L}_1$ implies ${L}_2 \ldist {L}_1$.
By transposition, we also have that ${L}_2 \not\ldist {L}_1$ implies ${L}_2 \not\lexpr {L}_1$.
The following is straightforward.

\begin{proposition}
For any strategy type $\mathcal{S}$ and model type $\mathcal{M}$, we have that
$\TATL_\mathcal{S}^\mathcal{M} \lexpr \STCTL_\mathcal{S}^\mathcal{M}$
(and thus $\TATL_\mathcal{S}^\mathcal{M} \ldist \STCTL_\mathcal{S}^\mathcal{M}$).
\end{proposition}
\begin{proof}
Follows as \TATL is a syntactic restriction of \STCTL.
\end{proof}

\begin{figure}[t]
  \centering
  \begin{minipage}{.27\textwidth}
    \begin{center}
    \scalebox{0.67}{
      \begin{tikzpicture}[>=stealth, node distance=1.5cm]
        \node[state, label=right:$\neg\prop{p}$] (q2) {$q_0$};
        \node[state, below left = of q2, label=right:$\neg\prop{p}$] (q3) {$q_1$};
        \node[state, below right = of q2, label=left:$\neg\prop{p}$] (q5) {$q_3$};
        \node[state, below right = of q3, label=right:$\prop{p}$] (q4) {$q_2$};

        \path[->] (q2) edge node[left] {$\alpha, \beta_{1}$} (q3);
        \path[->] (q3) edge node[left] {$\alpha_{2}, \beta$} (q4);
        \path[->] (q2) edge node[right] {$\alpha, \beta_{2}$} (q5);
        \path[->] (q5) edge node[right] {$\alpha_{2}, \beta$} (q4);

        \path[->] (q3) edge[loop left] node {$\alpha_{1}, \beta$} ();
        \path[->] (q5) edge[loop right] node {$\alpha_{1}, \beta$} ();
        \path[->] (q4) edge[loop below] node {$\alpha, \beta$} ();
      \end{tikzpicture}
    }
    \end{center}
  \end{minipage}
  \begin{minipage}{.2\textwidth}
    \begin{center}
      \scalebox{0.67}{
        \begin{tikzpicture}[>=stealth, node distance=1.5cm]
          \node[state, label=right:$\neg\prop{p}$] (q2) {$q'_0$};
          \node[state, below of = q2, label=right:$\neg\prop{p}$] (q3) {$q'_1$};
          \node[state, below of = q3, label=right:$\prop{p}$] (q4) {$q'_2$};

          \path[->] (q2) edge node[left] {\begin{tabular}{c}$\alpha, \beta_{1}$\\$\alpha, \beta_{2}$\end{tabular}} (q3);
          \path[->] (q3) edge node[left] {$\alpha_{2}, \beta$} (q4);

          \path[->] (q3) edge[loop left] node {$\alpha_{1}, \beta$} ();
          \path[->] (q4) edge[loop below] node {$\alpha, \beta$} ();
        \end{tikzpicture}
      }
    \end{center}
  \end{minipage}

\caption{Agent template $a$ (left) and agent $a'$ (right)}
\label{fig:expressivity}
\end{figure}
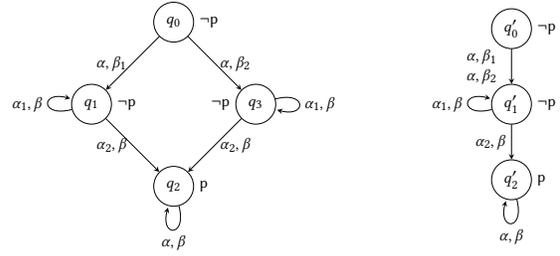

\begin{proposition}
For any strategy type $\mathcal{S}$ and model type $\mathcal{M}$, we have that
$\STCTL_\mathcal{S}^\mathcal{M} \not\ldist \TATL_\mathcal{S}^\mathcal{M}$
(and thus $\STCTL_\mathcal{S}^\mathcal{M} \not\lexpr \TATL_\mathcal{S}^\mathcal{M}$).
\end{proposition}
\begin{proof}
The proof is inspired by the proof of~\cite[Proposition~4]{Agotnes07irrevocable}.
Let us construct two multi-agent systems $S,S'$, each with $\Agents=\set{1,2}$.
Both agents in $S$ are based on the agent template $a$, depicted in Figure~\ref{fig:expressivity}~(left),
with the empty sets of clocks.
Note that the model $M$ of the system is isomorphic with the agent template, and the concrete model is identical with the model.

Similarly, both agents in $S'$ are based on the agent template $a'$, depicted in Figure~\ref{fig:expressivity}~(right),
again with no clocks.
The model $M'$ of the system is isomorphic with the agent template, and its concrete model identical with $M'$.
Moreover, $M$ and $M'$ are models with perfect information, in the sense that the local state of agent $1$ (resp.~$2$)
always uniquely identifies the global state in the model.
Thus, the sets of available strategies with perfect and imperfect information coincide,
and likewise of untimed vs.~timed strategies.
Furthermore, the strategic abilities for strategies with perfect vs.~imperfect recall are
the same for properties expressible in \ATL~\cite{Alur02ATL}.

It is easy to see that the pointed models $(M,q_0q_0)$ and $(M',q_0'q_0')$ are in alternating bisimulation~\cite{Alur98refinement},
and thus they satisfy exactly the same formulas of $\ATL[\ir]$.
By the above argument, they must satisfy the same formulas of $\TATL_\mathcal{S}^\mathcal{M}$, for any $\mathcal{M}\in\set{\cont,\untm}$ and all the strategy types $\mathcal{S}$ considered in this paper.
On the other hand, we have that the \SCTL (and hence also \STCTL) formula
$\varphi \equiv \coop{1}(\Epath\Sometm_{[0,\infty)}\prop{p} \land \Epath\Always_{[0,\infty)}\neg\prop{p})$
holds in $(M,q_0q_0)$ but not in $(M',q_0'q_0')$ for all the strategy types $\mathcal{S}$ and model types $\mathcal{M}\in\set{\cont,\untm}$.

For $\mathcal{M}=\disc$, we adapt the above construction by assuming that each transition consumes $1$ unit of time.
The models $M,M'$ of $S,S'$ are still isomorphic with $S,S'$, and their concrete models $CM,CM'$ are the tree-unfoldings of $M,M'$, thus they are alternating-bisimilar with $M,M'$~\cite{Agotnes07irrevocable}.
In consequence, they satisfy the same formulas of $\TATL_\mathcal{S}^\disc$, for all strategy types $\mathcal{S}$.
On the other hand, the above \SCTL and \STCTL formula $\varphi$ holds in $(M,q_0q_0)$ but not in $(M',q_0'q_0')$ for $\mathcal{M}=\disc$ and all $\mathcal{S}$.
\end{proof}

The following is a straightforward corollary.
\begin{theorem}
For any strategy type $\mathcal{S}$ and model type $\mathcal{M}$,
$\STCTL_\mathcal{S}^\mathcal{M}$ has strictly larger expressive and distinguishing power than $\TATL_\mathcal{S}^\mathcal{M}$.
\end{theorem}

\section{The Asynchronous Case}
\label{sec:async}

This section considers the case of asynchronous multi-agent systems (\AMAS),
providing the syntax and semantics of continuous time, discrete time, 
and untimed \AMAS.

\subsection{Asynchronous \MAS}

\emph{Asynchronous Multi-Agent Systems} (\AMAS~\cite{AAMASWJWPPDAM2018a}) are
a modern semantic model for the study of agents' strategies in asynchronous systems. 
Technically, \AMAS are similar to networks of automata that synchronise on
shared \event/s, and interleave local transitions to execute asynchronously
\cite{Fagin95knowledge,LomuscioPQ10a,AAMASWJWPPDAM2018a}.
However, to deal with agents coalitions, automata semantics (\eg for Timed Automata) must
resort to algorithms and additional attributes.
In contrast, by linking protocols to agents, \AMAS are a natural compositional formalism
to analyse multi-agent systems.

\subsection{Continuous Time \AMAS}
\label{sec:TAMAS}

\begin{definition}[\TAMAS]
\label{def:TAMAS:def}
A continuous time \AMAS (\TAMAS) is defined as \TMAS except for the following component:
\begin{itemize}
\item a (partial) \emph{local transition function} 
$\Trans_i: \Locations_i \times \Events_i \times \Constraintsi \times 2^{\Clocks_i} \fpart \Locations_i$
such that $\Trans_i(\loc_i,\evt,\cc,X) = \loc'_i$ for some $\loc'_i \in \Locations_i$  iff 
          $\evt \in \Prot_i(\loc_i)$, $\cc \in \Constraintsi$, and $X \subseteq \Clocks_i$;
\end{itemize}
\end{definition}
Note that as opposed to synchronous \MAS in \Cref{def:MAS:def}, the local transition
function of \AMAS is defined on local \event/s only.
This is also reflected in the formal definition of \AMAS models, 
also called Interleaved Interpreted Systems \cite{LomuscioPQ10,AAMASWJWPPDAM2018a}.

\begin{definition}[Model of a \TAMAS]
\label{def:TAMAS:model}
Let $\PV = \bigcup_{i=1}^n \PV_i$ be the union of the local propositions.
The \emph{model} of a \TAMAS is defined as the model for a \TMAS except for 
the following component:
\begin{itemize}
\item $\Trans: \States \times \Events \times \Constraints \times \Clocks \rightarrow \States$
	    \st
			$\Trans(\state,\evt,\bigwedge_{i \in \Agent(\evt)} \cc_i,\Clocks) = \state'$ iff
			$\forall i\in\Agent(\evt)$, $\Trans_i(\state^i,\evt,\cc_i,X_i) = \state'^i$,
			whereas	$\forall i\in\A\setminus\Agent(\evt)$, $\state^i = \state'^i$;
\end{itemize}
\end{definition}

\begin{definition}[\ACTS]
\label{def:TAMAS:semantics}
The \emph{concrete model} of a \TAMAS model is defined as the \CTS of a \TMAS model 
except for the component:
\begin{itemize}
\item $\to_c \subseteq \CStates \times (\Events \cup \Reals
_{0+}) \times \CStates$ is the \emph{transition relation},
defined by time- and action successors as follows:\\
$(\state,v) \xrightarrow[]{\delta}_c (\state,v+\delta)$
				for $\delta \in \Reals_{0+}$ and $v, v+\delta \in \llbracket \Invariant(\state) \rrbracket$,\\
$(\state,v) \xrightarrow[]{\evt}_c (\state',v')$
				iff there are $\evt \in \Events$, $\cc \in \Constraints$, $X \subseteq \Clocks$~s.t.:
				$\state \xrightarrow[]{\evt,\cc,X} \state' \in \Trans$,
				$v \in \llbracket\cc\rrbracket$,
				$v \in \llbracket\Invariant(\state)\rrbracket$,
				$v' = v[X := 0]$ 
				$v'~\in~\llbracket\Invariant(\state')\rrbracket$.
\end{itemize}
\end{definition}

\begin{example}
\label{ex:voters}
Consider the simple voting model in \Cref{fig:example-sync}, inspired by the election procedures in Estonia \cite{SpringallFD14}.
The voter (\V) needs to register first,
selecting one of the three voting modalities:
postal vote by mail ($\mathit{reg}_m$),
e-vote over the internet ($\mathit{reg}_i$),
or a traditional paper ballot at a polling station ($\mathit{reg}_p$).
The election authority (\EA) accepts \V's registration by synchronising with the registration transition. 
It then proceeds to send a voting package appropriate for \V's chosen modality
($\mathit{pack}_m$, $\mathit{pack}_i$ or $\mathit{pack}_p$),
\eg a postal ballot for voting by mail, e-voting access credentials, or the address of the local election office.
After receiving the package, \V casts a vote for the selected candidate ($
\mathit{vote1}_m$, etc.),
which is registered by \EA.
The local proposition \prop{v_i} denotes that \V voted for candidate $i$.

Time frames are associated with the voting process in the \EA automaton,
which accepts votes by mail between times 1 and 7,
by internet between 6 and 9,
and at the polling station between 10 and 11.
The ballot is closed at time 11.
Moreover, a voter must be registered for a modality before its respective voting period starts.

\begin{figure}[t]
  \centering
  \begin{minipage}{.27\textwidth}
    \begin{center}
    \scalebox{0.67}{
      \begin{tikzpicture}[>=stealth, node distance=2cm]
				\node[state, label=above:$\textbf{V}$,initial, initial text=] (q0) {};
				\node[state] at (-1.5,-1.5) (q1) {};
				\node[state] at (0,-1.5) (q2) {};
				\node[state] at (2,-1.5) (q3) {};
				\node[state] at (-2.5,-3) (q4) {};
				\node[state] at (0,-3) (q5) {};
				\node[state] at (2.5,-3) (q6) {};
				\node[state] at (-3.2,-4.5) (q4a) {$\prop{v_1}$};
				\node[state] at (-1.8,-4.5) (q4b) {$\prop{v_2}$};
				\node[state] at (-0.7,-4.5) (q5a) {$\prop{v_1}$};
				\node[state] at (0.7,-4.5) (q5b) {$\prop{v_2}$};
				\node[state] at (1.8,-4.5) (q6a) {$\prop{v_1}$};
				\node[state] at (3.2,-4.5) (q6b) {$\prop{v_2}$};

        \path[->] (q0) edge node[midway, sloped, above] {$\mathit{reg}_m$} (q1);
        \path[->] (q0) edge node[midway, sloped, above]
          {$\mathit{reg}_i$} (q2);
        \path[->] (q0) edge node[midway, sloped, above] {$\mathit{reg}_p$} (q3);
        \path[->] (q1) edge node[midway, sloped, above] {$\mathit{pack}_m$} (q4);
        \path[->] (q2) edge node[midway, sloped, above]
          {$\mathit{pack}_i$} (q5);
        \path[->] (q3) edge node[midway, sloped, above] {$\mathit{pack}_p$} (q6);
        \path[->] (q4) edge node[midway, sloped, above]
          {$\mathit{vote1}_m$} (q4a);
        \path[->] (q4) edge node[midway, sloped, above]
          {$\mathit{vote2}_m$} (q4b);
        \path[->] (q5) edge node[midway, sloped, above]
          {$\mathit{vote1}_i$} (q5a);
        \path[->] (q5) edge node[midway, sloped, above]
          {$\mathit{vote2}_i$} (q5b);
        \path[->] (q6) edge node[midway, sloped, above]
          {$\mathit{vote1}_p$} (q6a);
        \path[->] (q6) edge node[midway, sloped, above]
          {$\mathit{vote2}_p$} (q6b);

        \path[->] (q4a) edge[loop below] ();
        \path[->] (q4b) edge[loop below] ();
        \path[->] (q5a) edge[loop below] ();
        \path[->] (q5b) edge[loop below] ();
        \path[->] (q6a) edge[loop below] ();
        \path[->] (q6b) edge[loop below] ();
      \end{tikzpicture}
    }
    \end{center}
  \end{minipage}
  \begin{minipage}{.2\textwidth}
    \begin{center}
      \scalebox{0.67}{
        \begin{tikzpicture}[>=stealth, node distance=2cm]
          \node[state, label=above:$\textbf{EA}$] (q1i) {$x\leq0$};
          \node[state] at (-2.5,-1.3) (q1m) {$x\leq0$};
          \node[state] at (2.5,-1.3) (q1p) {$x\leq0$};
          \node[state,initial, initial text=] (q0) at (0,-3) {$t<=11$};
          \node[state] (q2) at (0,-5) {};

          \path[->] (q0) edge[bend left=10, left, near end]
            node {\begin{tabular}{c}$t\leq6$\\$\mathit{reg}_i$\\$x:=0$\end{tabular}}
            (q1i);
          \path[->] (q1i) edge[bend left=10, sloped, above, near start]
            node {$\mathit{pack}_i$} (q0);
          \path[->] (q0) edge[bend left=10, sloped, near end,below]
            node {\begin{tabular}{c}$t\leq1$\\$\mathit{reg}_m$\\$x:=0$\end{tabular}}
            (q1m);
          \path[->] (q1m) edge[bend left=10, sloped, near start,above]
            node {$\mathit{pack}_m$} (q0);
          \path[->] (q0) edge[bend left=10, sloped, near end, above]
            node {\begin{tabular}{c}$t\leq10$\\$\mathit{reg}_p$\\$x:=0$\end{tabular}}
            (q1p);
          \path[->] (q1p) edge[bend left=10, sloped, near start, below]
            node {$\mathit{pack}_p$} (q0);
          \path[->] (q0) edge
            node[left, sloped, midway, above] {\begin{tabular}{c}$t=11$\\$\mathit{close}$
            \end{tabular}} (q2);

          \path[->] (q0) edge[loop right, sloped, distance=2cm]
            node [above] {\begin{tabular}{c} $1\leq t\leq7, \mathit{vote1}_m$\\
              $1\leq t\leq7, \mathit{vote2}_m$ \end{tabular}} (q0);
          \path[->] (q0) edge[loop right, sloped]
            node {\begin{tabular}{c} $6\leq t\leq9, \mathit{vote1}_i$\\
              $6\leq t\leq9, \mathit{vote2}_i$ \end{tabular}} ();
          \path[->] (q0) edge[loop left, sloped]
            node {\begin{tabular}{c} $10\leq t\leq11, \mathit{vote1}_p$\\
              $10\leq t\leq11, \mathit{vote2}_p$ \end{tabular}} ();
        \end{tikzpicture}
      }
    \end{center}
  \end{minipage}

\caption{The \TAMAS of the voting scenario from \Cref{ex:voters}.}
\label{fig:example-sync}
\end{figure}
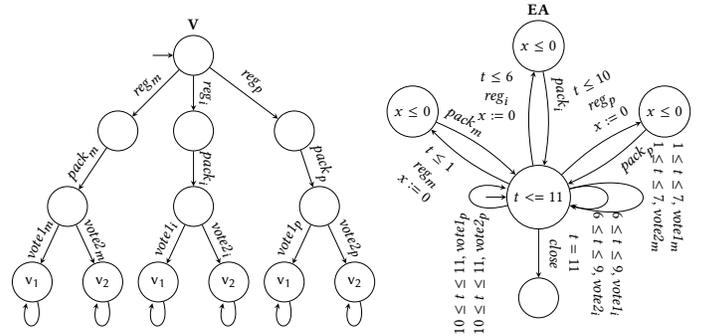

\end{example}

\subsection{Discrete Time \AMAS}

\emph{Discrete time asynchronous multi-agent systems} extend \AMAS with discrete time,
in a way similar to the synchronous case. 

\begin{definition}[\DAMAS]
\label{def:DAMAS:def}
A discrete time \AMAS (\DAMAS) is defined as \DMAS except for the following component: 
\begin{itemize}
\item a (partial) \emph{local transition function} $\Trans_i: \Locations_i \times \Events_i \fpart \Locations_i \times \mathbb{N}_{+}$
such that $\Trans_i(\loc_i,\evt)$ is defined iff $\evt \in \Prot_i(\loc_i)$;
\end{itemize}
\end{definition}

However, when agents share an action, the time the action takes is not enforced to be the 
same for all participants. 
Instead, the duration of the global action is the maximum of the participating agents' durations.
Thus, the slowest agent slows down its partners.

\begin{definition}[Model of \DAMAS]
\label{def:DAMAS:model}
The \emph{model} of a \DAMAS is defined as the model of a \DMAS
except for the component:
\begin{itemize}
\item $\Trans: \States\times \Events \rightarrow \States \times\mathbb{N}_{+}$ is the partial \textit{
transition function},
such that $\Trans(\state,\evt)= (\state',\max_{i \in\Agent(\evt)}\delta_i)$ iff $\Trans_i
(\state_i,\evt) = (\state'_i,\delta_i)$
for all $i \in \Agent(\evt)$, and $\state'_i=\state_i$ for all $i\in\Agents\setminus\Agent (\evt)$.
\end{itemize}
\end{definition}

These changes to local and global transitions are incorporated in the concrete \DAMAS model,
otherwise identical to that of a \DMAS.

\begin{definition}[\ADTS]
\label{def:DAMAS:semantics}
The \emph{concrete model} of a \DAMAS model is defined as the \DTS of a \DMAS model
except for the component:
\begin{itemize}
\item $\outm\colon\TStates\times\Events\to\TStates$ is a (partial)
transition function such that
$\outm((\state, d), \evt) = (\state', d+\delta) \text{ iff }
\Trans(\state, \evt) = (\state', \delta)$,
for $\state,\state'\in \States$, $\evt\in \Events$, 
$d \in\mathbb{N}$, and $\delta \in\mathbb{N}_{+}$.
\end{itemize}
\end{definition}

\subsection{Untimed \AMAS}

Untimed \AMAS can be defined as Timed \AMAS with no clocks, see below.
Note that the definition is essentially equivalent to the concept of 
an \emph{interleaved interpreted system} in~\cite{LomuscioPQ10}.
\begin{definition}[Untimed \AMAS]
\label{def:AMAS:def}
An \emph{untimed asynchronous multi-agent system}, 
simply \AMAS, is a \TAMAS with every $\Clocks_i = \emptyset$.
The model and concrete model of an \AMAS are equal and defined as in Definition~\ref{def:TAMAS:model} 
without clocks.
\end{definition}

\subsection{Model Checking in \AMAS}

The semantics of \STCTL (\TATL) is the same as in the synchronous case
except for each $\jevt$ to be replaced by $\evt$ in the paths.
In principle, the model checking procedures and complexity results for \STCTL and \TATL 
and their untimed variants given in \Cref{sec:modch} also apply to asynchronous models.
Note, however, that complexity is specified wrt. the model size,
which in \AMAS is significantly larger due to asynchronous interleaving of agents' actions.
On the other hand, the associated blow-up of state- and transition-space can be alleviated  
via techniques such as partial order reductions \cite{POR4ATL-JAIR}.

\section{Experiments}
\label{sec:expe}

In this section, we aim to show that model checking \STCTL[ir] is practically feasible.
To that end, we implemented the \TAMAS from \Cref{ex:voters} in the \imitator model checker \cite{AndreFKS12},
and conducted a set of initial experiments
using formulas $\varphi_{|\A|} = \coop{\A}\Epath\mathtt{F}_{[0,8]}\prop{v_1}$,
which specify that voter(s)\footnote{Coalition specified explicitly here for clarity; all voters in the \TAMAS are symmetrical.}
in $\A = \set{voter_1,\dots,voter_{|\A|}}$ have a strategy to vote for the first candidate 
within 8 time units,
\ie{}, reach a state labelled with the local proposition $\prop{v_1}$ before 8.

\imitator allows for \TCTL model checking and uses an asynchronous semantics on networks of timed automata, which fits our purposes.
Furthermore, as a state-of-the-art tool for Parametric Timed Automata,
it enables us to encode agents' strategies as parameters: for each coalition agent,
we add a parameter for each transition and a guard such that the
parameter corresponding to the transition is 1 while those corresponding to the other
transitions exiting the same location are 0. 
Note that this is not necessary when a single transition exits a location as there 
is no choice and thus no influence on the strategy.

Our model is scaled with the number of voters $v$ and the number of candidates $c$,
and we verify formulas $\varphi_1$, $\varphi_2$ and $\varphi_3$,
\ie, with one to three agents (voters) in the coalition $\A$ (\Cref{fig:results_all}, top).
The expected result is obtained: 
the voter(s) have a strategy to enforce $\Epath\mathtt{F}_{[0,8]}\prop{v_1}$,
which consists in choosing either the mail or the internet modality, and then voting for candidate 1.

While this already demonstrates the feasibility of \STCTL[ir] model checking,
the use of \imitator additionally provides (for free) the synthesis of \emph{all} 
strategies (\Cref{fig:results_all}, bottom).
However, this quickly faces a blowup in computation time.
On the other hand, a single strategy of one agent in the formula $\varphi_1$
can be obtained within the same timeout (120s) for significantly larger models,
with as many as 180 voters and 2 candidates, or 200 voters and 1 candidate.
The code and binaries required to replicate the experiments are accessible at \url{https://depot.lipn.univ-paris13.fr/mosart/publications/stctl}.

\begin{figure}[t]\centering
\begin{minipage}{.45\textwidth}\centering
		\scalebox{0.60}{
		\includegraphics[clip=true,trim=7.9cm 3.7cm 8.2cm 11.6cm]{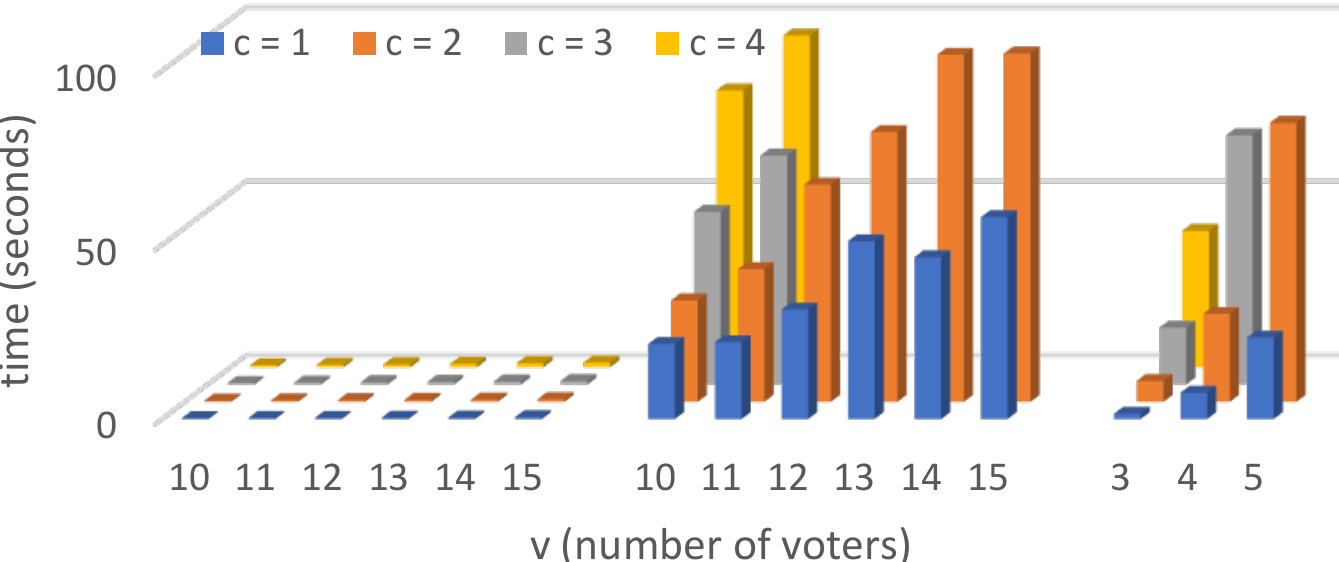}
		}
\end{minipage}
\hspace{.08\textwidth}
\begin{minipage}{.45\textwidth}\centering
		\scalebox{0.6}{
		\includegraphics[clip=true,trim=7.9cm 3.4cm 9.2cm 11.4cm]{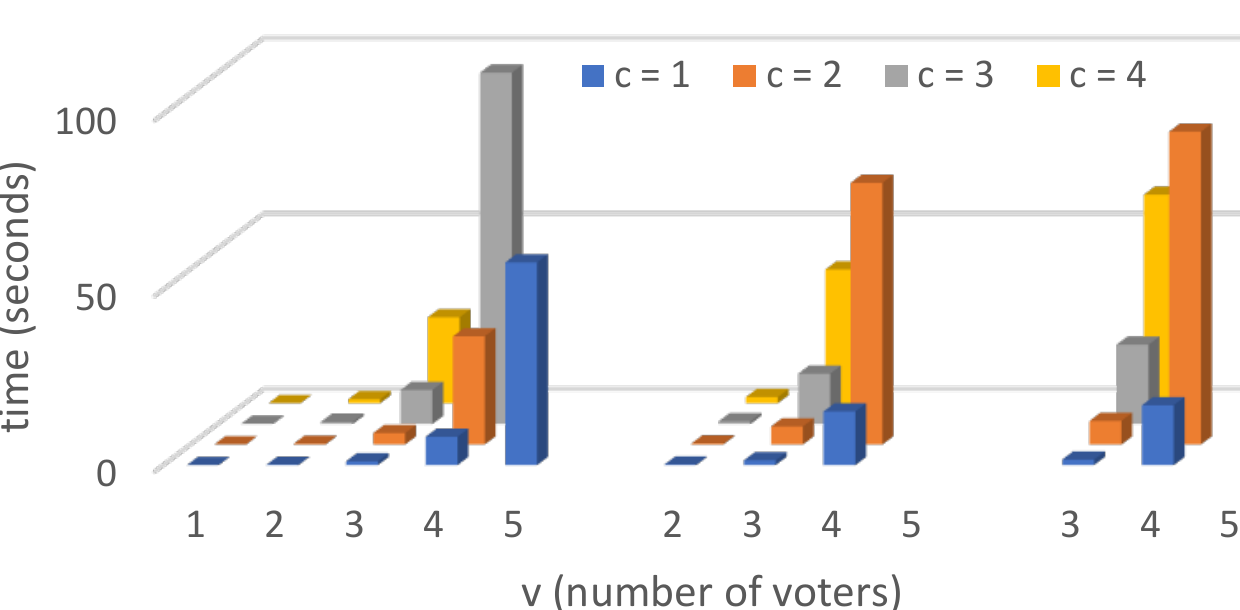}
		}
\end{minipage}
	\caption{Top: model checking $\varphi_1$ (left), $\varphi_2$ (middle), and $\varphi_3$ (right) with $v$ voters and $c$ candidates.
	Bottom: also synthesising all strategies.
	Missing bars indicate timeout (> 120s).}
	\label{fig:results_all}
\end{figure}

\section{Conclusions and Future Work}
\label{sec:conclu}

This paper shows that \STCTL, being a syntactic extension of \TATL, 
but interpreted over timed models with continuous semantics, in both synchronous and asynchronous settings,
is of theoretical and practical interest in model checking with ir- and Ir-strategies.
Our plans for future research include:
investigating also counting and timed strategies, 
a finer tuning of a model checking practical approach to easily 
capture all \STCTL properties, and
extending \STCTL to \STCTLs.
Moreover, since we have observed that synthesis of all strategies is too time consuming, 
but feasible even with the existing tool, we plan to implement a smarter, dedicated algorithm.


\begin{acks}
This work was partially funded by the CNRS IEA project MoSART
and by the PolLux/FNR projects STV (POLLUX-VII/1/2019) and SpaceVote.
\end{acks}


\balance
\bibliographystyle{ACM-Reference-Format}
\bibliography{report,wojtek,wojtek-own}


\end{document}